%% file: main.tex
\documentclass[11pt,a4paper]{article}

\usepackage[a4paper, total={5.8in, 9in}]{geometry}

\usepackage{nicefrac}
\usepackage{graphicx}         
\usepackage{amsmath,amsfonts}
\usepackage{amsthm}
\usepackage{amssymb}
\usepackage[english]{babel}
\usepackage[shortlabels]{enumitem}
\usepackage{mdframed}
\usepackage{afterpage}
\usepackage{xspace}
\usepackage{todonotes}
\usepackage{float}
\usepackage{soul}
\usepackage{complexity}
\usepackage{adjustbox}
\usepackage[ruled,linesnumbered,vlined]{algorithm2e}
\usepackage{epsfig}
\usepackage{bbm}
\usepackage{thm-restate}
\usepackage[colorlinks,linkcolor=Darkblue,filecolor=blue,citecolor=blue,urlcolor=Darkblue,pagebackref]{hyperref}
\usepackage[nameinlink]{cleveref}
\usepackage{framed}
\usepackage{xcolor} 
\hypersetup{
    colorlinks,
    linkcolor={blue},
    citecolor={blue},
    urlcolor={blue}
}

\newtheorem{theorem}{Theorem}[section]
\newtheorem{lemma}[theorem]{Lemma}

\newtheorem{proposition}[theorem]{Proposition}

\newtheorem{claim}{Claim}

\numberwithin{theorem}{section}

\def\eps{\varepsilon}
\def\poly{\mathrm{poly}}

\def\N{{\mathbb N}}

\def\calB{{\cal B}}

\def\calI{{\cal I}}
\def\1{\mathbb{I}}

\def\xx{{\textrm{x}}}

\definecolor{myyellow}{RGB}{240,217,1}
\definecolor{mygreen}{RGB}{143,188,103}
\definecolor{myred}{RGB}{234,38,40}
\definecolor{myblue}{RGB}{53,101,167}

\renewcommand{\R}{\mathbb{R}}
\renewcommand{\E}{\mathbb{E}}

\newcommand{\spann}{\mathrm{span}}
\newcommand{\argmin}{\mathrm{argmin}}

\newcommand{\frakB}{\mathfrak{B}}

\newcommand{\calS}{\mathcal{S}}
\newcommand{\frakI}{\mathfrak{I}}
\newcommand{\frakJ}{\mathfrak{J}}

\newcommand{\lra}{$\ell_0$-\textsc{Low Rank Approximation}\xspace}
\newcommand{\lraa}{$\ell_0$-\textsc{LRA}\xspace}

\newcommand{\rlra}{\textsc{Restricted} $\ell_0$-\textsc{Low Rank Approximation}\xspace}

\makeatletter
\let\@oldnoitemerr\@noitemerr %Save the command definition                      
\newcommand\noitemerroroff{\let\@noitemerr\relax}
\newcommand\noitemerroron{\let\@noitemerr\@oldnoitemerr}
\makeatother

\title{A PTAS for $\ell_0$-Low Rank Approximation: \\ Solving Dense CSPs over Reals}

\author{Vincent Cohen-Addad \thanks{Google Research, France} \and Chenglin Fan \thanks{Sorbonne Université, France} \and Suprovat Ghoshal\thanks{Northwestern University, USA} \and Euiwoong Lee\thanks{University of Michigan, USA} \and Arnaud de Mesmay\thanks{LIGM, CNRS, Univ. Gustave Eiffel, ESIEE Paris, F-77454 Marne-la-Vallée, France} \and Alantha Newman\thanks{Laboratoire G-SCOP (CNRS, Grenoble-INP), France} \and Tony Chang Wang\thanks{University of Wisconsin-Madison, USA}}

\begin{document}

\date{}

\maketitle

\begin{abstract}

We consider the \lra problem, where the input consists of a matrix $A \in \R^{n_R \times n_C}$ and an integer $k$, and the goal is to find a matrix $B$ of rank at most $k$ that minimizes $\| A - B \|_0$, which is the number of entries where $A$ and $B$ differ. 
For any constant $k$ and $\eps > 0$, we present a polynomial time $(1 + \eps)$-approximation time for this problem, which significantly improves the previous best $\poly(k)$-approximation. 

Our algorithm is obtained by viewing the problem as a Constraint Satisfaction Problem (CSP) where each row and column becomes a variable that can have a value from $\R^k$. 
In this view, we have a constraint between each row and column, which results in a {\em dense} CSP, a well-studied topic in approximation algorithms. 
While most of previous algorithms focus on finite-size (or constant-size) domains and involve an exhaustive enumeration over the entire domain, 
we present a new framework that bypasses such an enumeration in $\R^k$. 
We also use tools from the rich literature of Low Rank Approximation in different objectives (e.g., $\ell_p$ with $p \in (0, \infty)$) or domains (e.g., finite fields/generalized Boolean).
We believe that our techniques might be useful to study other real-valued CSPs and matrix optimization problems.

On the hardness side, when $k$ is part of the input, we prove that \lra is NP-hard to approximate within a factor of $\Omega(\log n)$.
This is the first superconstant NP-hardness of approximation for any $p \in [0, \infty]$ that does not rely on stronger conjectures (e.g., the Small Set Expansion Hypothesis). 
\end{abstract}

\section{Introduction}
Computing a low rank approximation of a given matrix is one of the most fundamental algorithmic tasks in data analysis and machine learning. 
Formally, given a matrix $A \in \R^{n_R \times n_C}$ and an integer $k \in \N$, the goal is to compute a matrix $B \in \R^{n_R \times n_C}$ of rank at most $k$ that minimizes some distance measure between $A$ and $B$.
The Frobenius norm $\| A - B \|_F = (\sum_{i, j} (A_{i, j} - B_{i,j})^2)^{1/2}$ and its generalizations to Schatten norms can be optimized in polynomial time for any $k$, and there is a rich literature on faster algorithms to compute them, possibly for special classes of matrices or more restricted models of computations~\cite{clarkson2017low, clarkson2017lowsoda, musco2017sublinear, jiang2021single, bakshi2020robust, li2020input, woodruff2022improved, bakshi2022low}.
Many variants that are not expected to have a polynomial time exact algorithm have been actively studied as well, including tensor versions~\cite{song2019relative} or weighted versions~\cite{razenshteyn2016weighted, ban2019regularized} where each entry has different weights. 
The (entrywise) $\ell_p$ objective 
$\| A - B \|_p := (\sum_{i, j} (A_{i, j} - B_{i,j})^p)^{1/p}$ is another generalization of the Frobenius norm~\cite{chierichetti2017algorithms, song2017low, ban2019ptas, mahankali2021optimal}. 
This objective with $0 \leq p < 2$ is generally considered to be more robust than the Frobenius norm, because in the Frobenius norm, a few outlier entries (whose values are very far from correct) can have a large effect on the other entries in a solution.

In this paper, we focus on \lra (\lraa), where $\| A - B \|_0$ is defined to be the number of entries where $A$ and $B$ differ. It is a maximally robust objective function in the aforementioned sense, which was used in the notion of robust PCA~\cite{candes2011robust}. The choice of the $\ell_0$ metric also makes particular sense in contexts where the data is not endowed with a natural underlying metric. 
The \lraa problem coincides with the matrix rigidity problem over the reals, which has been studied in the context of complexity theory~\cite{grigoriev1976using, valiant1977graph} and parameterized complexity~\cite{fomin2017rigidity}, and is closely related to matrix completion~\cite{johnson1990matrix, candes2010matrix, keshavan2010matrix, recht2011simpler}.
The special case when $A \in \{0, 1\}^{m \times n}$, which is NP-hard even for $k = 1$~\cite{gillis2018complexity, dan2018low}, has been also well studied~\cite{fomin2019approximation, fomin2020parameterized}. 
Another related problem is the {\em Metric Violation Distance (MVD)} problem where the input is a symmetric matrix $A \in \R^{n \times n}$ and the goal is to find a matrix $B \in \R^{n \times n}$ representing the pairwise distances in a metric space to minimize $\| A - B \|_0$~\cite{fan2018metric, cohen2022fitting}.

The previous best approximation algorithms for \lraa, given by Bringmann, Kolev, and Woodruff~\cite{BKW17}, achieve an $O(k^2)$-approximation in time $n^{O(k)}$ and an $(2 + \eps)$-approximation when $k = 1$. (Let $n := \max(n_R, n_C)$.) 
It is in stark contract to $(1+\eps)$-approximations for fixed $k$ when $p \in (0, 2)$ or the domain is constant-size~\cite{ban2019ptas, fomin2019approximation}. 
We bridge this gap, showing that \lraa admits a PTAS for every constant $k$ as well. 

\begin{restatable}{theorem}{min}\label{thm:main}

For any fixed constants $k \in \N$ and $\eps > 0$, there exists a $(1+\eps)$-approximation algorithm for \lra that runs in time $n^{2^{\poly(k/\eps)}} \poly(\tau)$, where $\tau$ is an upper bound on the bitsize of the coefficients of the input matrix. 
\end{restatable}

Our algorithm works by computing matrices $U \in \R^{n_R \times k}$ and $W \in \R^{n_C \times k}$ so that $B=UW^T$. Each entry of our solution $U$, $W$ and $B$ is not guaranteed to be a rational number and will be described by the {\em Thom encoding}; roughly, it will be the unique solution to a system of an $2^{\poly(k/\eps)}$ polynomial (in)equalities whose coefficients have bit complexity at most $\poly(\tau)$. See Section~\ref{sec:algeq} for more background.

Our result is inspired by the connection between Low Rank Approximation
and the
well-studied topic of {\em dense Constraint Satisfaction Problems (CSPs)}. 
In the context of \lraa, we consider each row and column as a variable that can have a value from $\R^k$. For each row $i$ and column $j$, 
we have a constraint that is satisfied if $\langle u_i, v_j \rangle = A_{i, j}$, where $u_i$ and $v_j$ denote the vectors chosen by $i$ and $j$ respectively. 
Dense CSPs are a central topic in approximation algorithms, and there are PTASes using various methods (e.g., sampling~\cite{arora1995polynomial, alon2003random, de2005tensor, mathieu2008yet, karpinski2009linear, barak2011subsampling, yaroslavtsev2014going, manurangsi2015approximating}, regularity lemma~\cite{frieze1996regularity, coja2010efficient}, convex hierarchies~\cite{de2007linear,
barak2011rounding, guruswami2011lasserre, yoshida2014approximation}).

However, all previous techniques crucially rely on the fact that the {\em domain} of a CSP is finite (and bounded as a function of $n$), which makes it nontrivial to apply these ideas to \lraa.
We overcome such a difficulty by introducing a new framework that allows us to use tools from both the Low Rank Approximation and dense CSP literatures (see Section~\ref{sec:techniques} for more detailed description of our techniques). We hope that it might be useful for other matrix problems and CSPs with real domains. On our way to proving Theorem~\ref{thm:main}, our first result is to provide an $\eps$-\emph{additive} approximation algorithm, 
which returns a solution $B$ with the guarantee that 
$\| A - B \|_0$ is at most the optimal value plus $\eps n_R n_C$.

\begin{restatable}{theorem}{additive}\label{thm:additive}
For any fixed constants $k \in \N$ and $\eps>0$, there exists an algorithm computing an $\varepsilon$-additive approximation to \lraa that runs in time $n^{(1/\eps)^{\poly(k)}} \poly(\tau)$, where $\tau$ is an upper bound on the bitsize of the coefficients of the input matrix.
\end{restatable}

On the hardness side, when $k$ is part of the input, we prove that \lra is NP-hard to approximate within a factor of $\Omega(\log n)$, which implies that the superpolynomial dependence on $k$ is necessary.
To the best of our knowledge, this is the first superconstant {\em NP-hardness} of approximation for any $p \in [0, \infty]$ that does not rely on stronger conjectures. 
The only known $\omega(1)$-hardness of $\ell_p$-\textsc{Low Rank Approximation}, which holds for $p \in (1, 2)$, relies on the Small Set Expansion Hypothesis~\cite{ban2019ptas}. 

\begin{restatable}{theorem}{hardness}\label{thm:hardness}
When $k$ is part of the input, it is NP-hard to approximate \lra within a factor of $\Omega(\log n)$. 
\end{restatable}

Theorem~\ref{thm:main} features a doubly-exponential dependency on the parameters $k$ and $\varepsilon$, and this dependency is not fixed-parameter tractable. Furthermore, our algorithm heavily relies on algorithms from real algebraic geometry that quickly become impractical. Whether one can improve this complexity, both from a theoretical and practical point of view, is the main question arising from our work. It would also be interesting to adapt our framework to other problems which can be phrased as CSPs over the reals, such as other matrix factorization problems, or finite-dimensional versions of problems involving distances (see, e.g., ~\cite{fan2018metric, cohen2022fitting}).

\subsection{Techniques}
\label{sec:techniques}

In this section, we give an overview of the techniques involved in the proof of Theorems~\ref{thm:main} and~\ref{thm:additive}. Let $A \in \R^{n_R \times n_C}$ be the input matrix, let $U W^T$ be an optimal solution where $U \in \R^{n_R \times k}$ and $W \in \R^{n_C \times k}$. We denote by $u_i$ the $i$th row of $U$ and $w_i$ the $i$th row of $W$ (as column vectors). Let $OPT = | \{ (i, j) : A_{i, j} \neq \langle u_i, w_j \rangle \} |$ be the number of {\em errors} that the optimal solution makes.

Both algorithms rely on PTASes for constraint satisfaction problems, which we first introduce. A Constraint Satisfaction Problem of arity $2$ (2-CSP) consists of (i) a family of $n$ variables $V$, which can take values within a given alphabet $D$ (also called domain), and (ii) a family of constraints $C$ between some pairs of variables, where each constraint is a subset of $D \times D$. An assignment is a map $\varphi: V \rightarrow D$. The set of pairs of variables between which there is a constraint is encoded in a graph $G$, called the primal graph (or Gaifman graph) of the 2-CSP. The goal is to find an assignment that maximizes the number of satisfied constraints (Max-2-CSP), or that minimizes the number of unsatisfied constraints (Min-2-CSP). While these two problems are obviously equivalent in the realm of exact algorithms, providing approximation algorithms leads to different challenges in the minimization and the maximization setting. In a nutshell, efficiently approximating a Min-2-CSP requires performing very well on the instances where almost all the constraints are satisfiable, while approximating a Max-2-CSP requires performing well in the opposite regime, when a very small number of constraints is satisfiable. This explains why when the graph is dense and the alphabet size is of constant size, it is much easier to obtain a PTAS for the Max-CSP: in this regime, the maximum number of satisfiable constraints is $\Omega(n^2)$, as can easily be proved by taking a random assignment. Therefore, in order to design a PTAS for a dense Max-2-CSP and a constant-size alphabet, it suffices to devise a $\varepsilon$-additive approximation, akin to the one we are aiming for in Theorem~\ref{thm:additive}. 

Our approach to prove Theorems~\ref{thm:main} and~\ref{thm:additive} is to formulate the \lra problem as 2-CSP, where the primal graph is the complete bipartite graph. There is one variable for each row and each column, the alphabet is $\mathbb{R}^k$ and the constraint between the row $i$ and the column $j$ is $\langle u_i,w_j \rangle=A_{i,j}$. Then we would like to use the PTASes for such dense CSPs available in the literature, but the key issue is that our alphabet size is infinite. Therefore, the most technical part in both our algorithms consists in reducing the alphabet size to a constant 
size: computing for each row $i$ and column $j$ a constant-size alphabet $\Sigma_i$ or $\Sigma_j$ of vectors in $\mathbb{R}^k$, such that there exists a near-optimal solution using exclusively vectors from these alphabets. Throughout this overview, whenever we refer to ``constant'', the constant depends on $k$ and $\varepsilon$; we refer to the proofs for the precise values.

\subsubsection{Additive approximation scheme: Theorem~\ref{thm:additive}}\label{SS:additive}

A classical approach to design additive approximation schemes for dense CSPs is to sample a constant number of variables~\cite{mathieu2008yet,yaroslavtsev2014going}, guess their values and then extrapolate from this sample the values of all the other variables. Since our domain size is infinite, we cannot guess the values here, and instead our key contribution is to prove the existence of a constant-size set of variables and constraints among them {\em beyond the ones given by matrix entries} such that \emph{any} solution to the constraints between those can be extended to be a near-optimal solution on the full set of variables. Such a solution can be computed using real algebraic solvers~\cite{roy}. This idea might be of independent interest to other constraint satisfaction problems over the reals.

We now get into more details. The entire algorithm behind Theorem~\ref{thm:additive} is outlined in Figure~\ref{A:mainalgo}. For simplicity, we assume in this section that $n:=n_R=n_C$.

\paragraph*{The rigid case.} We first explain the intuition behind it by investigating a particular case. We first assume that there exists an optimal solution $U,W$ that is \emph{rigid}: every $k \times k$ submatrix of $U$ and $W$ has full rank (recall that $k$ is the target rank in our \lra problem). If at most $\varepsilon n^2$ constraints are satisfied, any solution is an $\varepsilon$-additive approximation. Otherwise, we consider the bipartite graph $G$, where the vertices are the rows and columns of $A$, and there is an edge whenever the constraint between row $i$ and column $j$ is satisfied, i.e., $\langle u_i,w_j \rangle=A_{i,j}$. Since this graph has at least $\varepsilon n^2$ edges, the Kovari-S\`os-Turan theorem~\cite{kst} guarantees that it admits a complete bipartite subgraph $G':=K_{k,k}$ as a subgraph. This complete bipartite subgraph enforces a solution on the corresponding rows of $U$ and columns of $W$, which is unique up to the natural symmetries of the problem. More formally, by our rigidness assumption, up to\footnote{This does not change the value of the solution since one can change $W$ accordingly: $(UC)(C^{-1}W^T)=UW^T$.} applying an invertible matrix $C$ to $U$ we can assume that the submatrix of $U$ induced by the rows of $G'$ is the identity, and then $W$ must exactly match the submatrix of $A$ induced by $G'$. Then, every row $u_i$ that is adjacent in $G$ to all the columns of $G'$ (thus forming a $K_{k+1,k}$) also has its value completely determined by the constraints of $G'$, since it is the unique solution to a linear system of full rank.

The Kovari-S\`os-Turan theorem can be strengthened to a supersaturated version, showing that (Lemmas~\ref{lem:kst1} and~\ref{lem:kst}) not only there exists a $K_{k+1,k+1}$ subgraph, but there are a lot of them, and actually most edges of $G$ belong to many of them. This suggests the following algorithm. First, we sample a constant number of columns and vertices uniformly at random, and we guess the subgraph $G'$ of $G$ induced by this subset, which we call the \emph{core} of the solution. Now, let us assume that we can compute a family of rows and columns for this constant-size core {\em that exactly matches the optimal solution} (perhaps modulo the natural symmetries of the problem). Then, for any edge $(i, j)$ not in $G'$ (except for a negligible portion of those), we can prove that it belongs in $G$ to a $K_{k+1,k+1}$, where the other $2k$ vertices are in $G'$. Therefore, by the rigidness assumption, the rows and columns of $G'$ induce a unique solution for $i$ and $j$, which thus matches the optimal solution. We can thus define for a vertex $v$ an alphabet $\Sigma_v$ as being, for each possible choice of $K_{k,k}$ in $G'$ that $v$ could be adjacent to, the unique solution that it induces for $v$. Since this alphabet has constant size for each vertex, we can now appeal to standard Max-2-CSPs algorithms~\cite{yaroslavtsev2014going} to obtain the required $\eps$-additive approximation.

This algorithm requires us to compute the restriction of an optimal solution to the constant-size subset of rows and columns induced by $G'$. Such a solution must satisfy a family of quadratic equations: for all $(i, j) \in G'$, we should have $\langle u_i,v_j \rangle= A_{i,j}$, where $u_i$ and $v_j$ are unknown vectors in $\mathbb{R}^k$. We can solve such systems of equations using algorithms from real algebraic geometry, which more generally can be used to solve\footnote{There are subtle issues involving what it means to ``solve'' such a system of equations, see the discussion in Section~\ref{sec:algeq}.} any polynomial system of (in)-equations (or even any problem in the Existential Theory of the Reals, see Section~\ref{sec:algeq}) in exponential time. Such algebraic solvers have already been used in multiple algorithms in Low Rank Approximation and its variants, see for example~\cite{fomin2020parameterized,razenshteyn2016weighted,song2019relative}, but one key difference is that in~\cite{razenshteyn2016weighted,song2019relative}, they were used to optimize the objective function (which was itself polynomial) over a sketch. This is not possible for us because we use the $\ell_0$-norm. Since in our case, the systems have constant size, we can afford to pay the exponential complexity. However, a key issue appears: in contrast to the case where the entire $G'$ was equal to $K_{k, k}$, 
in general $G'$ can be an arbitrary graph so that such a solution will in general not be unique, even after quotienting by the natural symmetries of the problem. Therefore, it could be that the solution that we compute on $G'$ is fundamentally different from the optimal one, thus leading to alphabets which do not contain an optimal solution. 

\begin{figure}[t]
  \usetikzlibrary{matrix,positioning}
  \adjustbox{width=\textwidth,keepaspectratio}{ 
\begin{tikzpicture}[
mymatrix/.style={
  matrix of math nodes,
  outer sep=0pt,
  nodes={
    draw,
    anchor=base,
    text depth=.5ex,
    text height=2ex,
    text width=2em,
    align=center,
    text=gray
  },
  left delimiter=[,
  right delimiter=],
  }
]
%the matrices
\matrix[mymatrix] (A)
{
|[text=black]|1 & |[text=black]|0 \\
|[text=black]|0 & |[text=black]|1 \\
|[text=black]|1 & |[text=black]|0 \\
|[text=black]|0 & |[text=black]|1 \\
|[text=black]|u_{5,1} & |[text=black]|u_{5,2} \\
};
\matrix[mymatrix,right=of A.north east,anchor=north west] (prod)
{
 |[text=black]|2& |[text=black]|1 & \xx &\xx\\
 |[text=black]|1& |[text=black]|0 & \xx & \xx\\
 \xx& \xx& |[text=black]|0 & |[text=black]|2\\
 \xx& \xx& |[text=black]|2 & |[text=black]|0\\
 |[text=black]|3& |[text=black]|1& |[text=black]|1&|[text=black]|1\\
};
\matrix[mymatrix,above=of prod.north west,anchor=south west] (B)
{
 |[text=black]|2 & |[text=black]|1 &  |[text=black]|0 & |[text=black]|2\\
 |[text=black]|1 & |[text=black]|0 &  |[text=black]|2 & |[text=black]|0\\
};

%the labels for the matrices
\node[font=\large,left=10pt of A] {$U$};
\node[font=\large,above=2pt of B] {$W^T$};

%the frames in both matrices
\draw[myyellow,line width=2pt]
  ([shift={(1.2pt,-1.2pt)}]A-1-1.north west) 
  rectangle 
  ([shift={(-1.2pt,1.2pt)}]A-2-2.south east);
  \draw[myblue,line width=2pt]
  ([shift={(1.2pt,-1.2pt)}]A-5-1.north west) 
  rectangle 
  ([shift={(-1.2pt,1.2pt)}]A-5-2.south east);
\draw[myyellow,line width=2pt]
  ([shift={(1.2pt,-1.2pt)}]B-1-1.north west) 
  rectangle 
  ([shift={(-1.2pt,1.2pt)}]B-2-2.south east);
  \draw[mygreen,line width=2pt]
  ([shift={(1.2pt,-1.2pt)}]B-1-3.north west) 
  rectangle 
  ([shift={(-1.2pt,1.2pt)}]B-2-4.south east);
\draw[myyellow,line width=2pt]
  ([shift={(1.2pt,-1.2pt)}]prod-1-1.north west) 
  rectangle 
  ([shift={(-1.2pt,1.2pt)}]prod-2-2.south east);
 
\draw[mygreen,line width=2pt]
  ([shift={(1.2pt,-1.2pt)}]A-3-1.north west) 
  rectangle 
  ([shift={(-1.2pt,1.2pt)}]A-4-2.south east);
\draw[myblue,line width=2pt]
  ([shift={(1.2pt,-1.2pt)}]prod-5-1.north west) 
  rectangle 
  ([shift={(-1.2pt,1.2pt)}]prod-5-4.south east);

   \draw[mygreen,line width=2pt]
   ([shift={(1.2pt,-1.2pt)}]prod-3-3.north west) 
   rectangle 
   ([shift={(-1.2pt,1.2pt)}]prod-4-4.south east);

\end{tikzpicture}

\hspace{1cm}

\begin{tikzpicture}[
mymatrix/.style={
  matrix of math nodes,
  outer sep=0pt,
  nodes={
    draw,
    anchor=base,
    text depth=.5ex,
    text height=2ex,
    text width=2em,
    align=center,
    text=gray
  },
   left delimiter=[,
  right delimiter=],
  }
]
%the matrices
\matrix[mymatrix] (A)
{
|[text=black]|1 & |[text=black]|0 \\
|[text=black]|0 & |[text=black]|1 \\
|[text=black]|2 & |[text=black]|0 \\
|[text=black]|0 & |[text=black]|2 \\
|[text=black]|1 & |[text=black]|1 \\
};
\matrix[mymatrix,right=of A.north east,anchor=north west] (prod)
{
 |[text=black]|2& |[text=black]|1 & \xx &\xx\\
 |[text=black]|1& |[text=black]|0 & \xx & \xx\\
 \xx& \xx& |[text=black]|0 & |[text=black]|2\\
 \xx& \xx& |[text=black]|2 & |[text=black]|0\\
 |[text=black]|3& |[text=black]|1& |[text=black]|1&|[text=black]|1\\
};
\matrix[mymatrix,above=of prod.north west,anchor=south west] (B)
{
 |[text=black]|2 & |[text=black]|1 &  |[text=black]|0 & |[text=black]|1\\
 |[text=black]|1 & |[text=black]|0 &  |[text=black]|1 & |[text=black]|0\\
};

%the labels for the matrices
\node[font=\large,left=10pt of A] {$U$};
\node[font=\large,above=2pt of B] {$W^T$};

%the frames in both matrices
\draw[myyellow,line width=2pt]
  ([shift={(1.2pt,-1.2pt)}]A-1-1.north west) 
  rectangle 
  ([shift={(-1.2pt,1.2pt)}]A-2-2.south east);
  \draw[myblue,line width=2pt]
  ([shift={(1.2pt,-1.2pt)}]A-5-1.north west) 
  rectangle 
  ([shift={(-1.2pt,1.2pt)}]A-5-2.south east);
\draw[myyellow,line width=2pt]
  ([shift={(1.2pt,-1.2pt)}]B-1-1.north west) 
  rectangle 
  ([shift={(-1.2pt,1.2pt)}]B-2-2.south east);
  \draw[mygreen,line width=2pt]
  ([shift={(1.2pt,-1.2pt)}]B-1-3.north west) 
  rectangle 
  ([shift={(-1.2pt,1.2pt)}]B-2-4.south east);
\draw[myyellow,line width=2pt]
  ([shift={(1.2pt,-1.2pt)}]prod-1-1.north west) 
  rectangle 
  ([shift={(-1.2pt,1.2pt)}]prod-2-2.south east);
 
\draw[mygreen,line width=2pt]
  ([shift={(1.2pt,-1.2pt)}]A-3-1.north west) 
  rectangle 
  ([shift={(-1.2pt,1.2pt)}]A-4-2.south east);
\draw[myblue,line width=2pt]
  ([shift={(1.2pt,-1.2pt)}]prod-5-1.north west) 
  rectangle 
  ([shift={(-1.2pt,1.2pt)}]prod-5-4.south east);

   \draw[mygreen,line width=2pt]
   ([shift={(1.2pt,-1.2pt)}]prod-3-3.north west) 
   rectangle 
   ([shift={(-1.2pt,1.2pt)}]prod-4-4.south east);

\end{tikzpicture}
}
  \caption{Adding rows to the core in order to control its inner dependencies. The x values are not in the graph $G$ of the optimal solution. Left: Choosing arbitrary solutions for the yellow and the green $K_{2,2}$'s might lead to an unsolvable system of equations for the fifth row (blue). Right: Adding the fifth row to the core synchronizes the yellow and green $K_{2,2}$'s. }
  \label{F:sync}
\end{figure}

We solve this issue by adding additional data in a non-random way to the core $G'$. We explain the main idea on a simple instance, which is illustrated in Figure~\ref{F:sync} with $k=2$. Suppose that $G'$ consists of two vertex-disjoint $K_{k,k}$ subgraphs of $G$. Solving the corresponding system of equations and taking an arbitrary solution would lead to vectors for rows and columns 
which are completely uncorrelated between the two subgraphs. If some row $i$ not in $G'$ is adjacent in $G$ to all the columns of $G'$, the values of the row $i$ suggested by the two subgraphs will therefore never match. In such a case, we add the row $i$ to the core, yielding a \emph{supercore}. When we solve the system of polynomial equations on this supercore, the added row will have the effect of correlating the solutions on the two subgraphs of $G'$. Of course, we should not add all such rows to the core, since we want the supercore to also be constant-size, but our framework shows that it suffices to add a constant number of such rows and columns in order to account for all the required correlations between the various parts of the core in the optimal solution. We emphasize that the rows and columns added in the supercore cannot in general be chosen randomly. Thus our algorithm does not actually proceed by sampling and requires enumerating all the subsets up to some constant size: this is the reason for the $n^{f(k,\varepsilon)}$ complexity of our algorithm in Theorem~\ref{thm:additive}, as opposed to the FPT complexity of most PTASes in the literature for dense Max-CSPs.

\paragraph*{Extending to the general case.} In the general case, we cannot assume that there is an optimal solution that is rigid. 
In that setting, our algorithm still starts by guessing a supercore $G'$ and solving the corresponding system of polynomial equations. However, such a supercore will not in general induce a unique solution for a row or a column not contained in it, even if it is fully adjacent to a $K_{k,k}$ in $G'$, since the corresponding linear system of equations may not be full rank: this poses an issue when defining the alphabets. If one takes an arbitrary solution, then for a row $i$ and a column $j$ that are not in the supercore but form an edge $(i, j)$ in $G$, even if one guesses correctly to which $K_{k,k}$'s $i$ and $j$ are attached in $G$, the vectors $u_i$ and $w_j$ will in general not come from an optimal solution. Thus there is no guarantee that the constraint $(i, j)$ will be satisfied.

We solve this issue by adding even more data to the supercore. Since it has constant-size, we can afford to guess the entire system of linear dependencies between its elements, as this is encoded in a combinatorial object called a matroid~\cite{oxley}. Then, for vertices not in the supercore, we include in the alphabet not only to which parts of the core they are attached, but also with which of the independent sets of this matroid they are dependent in the optimal solution. This information is encompassed in submatrices that we call \emph{pieces}, which are also constantly many. Then we prove that guessing correctly the information of which edge belongs to which pieces will suffice to ensure correct edges, even outside of the supercore; this follows from an easy linear algebraic Lemma~\ref{lem:linear}.

We emphasize that due to the infinite size of the alphabet, our additive approximation scheme in Theorem~\ref{thm:additive} does not readily provide a PTAS for the maximization version of $\ell_0$-\textsc{Low Rank Approximation}, since there is in general no $\Omega(n^2)$ lower bound for the value of this CSP despite the density of the primal graph. We leave the existence of such a PTAS as an open question.

\subsubsection{Multiplicative approximation scheme: Theorem~\ref{thm:main}}
\label{sec:techniques-min}
\input{techniques-min.tex}

\section{Preliminaries}\label{sec:algeq}

\paragraph{Algorithms for real semialgebraic sets.} 
Our algorithm for Theorem~\ref{thm:additive} makes heavy use of algorithms of real algebraic geometry to solve systems of polynomial equations over the reals, which we use as a black-box. We refer to the book of Basu, Pollack and Roy~\cite{roy} for all the necessary background on this topic and highlight here the precise results that we rely on. A \emph{semialgebraic set} is the set of solutions to a system of polynomial equations or inequations over the reals, or any finite union of such sets. The \emph{projection} of a semialgebraic set $X \subseteq \mathbb{R}^{n_1+n_2}$ on a linear subspace $\mathbb{R}^{n_1}$ is the set $\{x_1, \ldots , x_{n_1} \mid \exists x_{n_1+1}, \ldots, x_{n_2} \textrm{ such that } x_1, \ldots , x_{n_2} \in X\}$. A consequence of the Tarski-Seidenberg theorem (see, e.g.,~\cite[Theorem~2.77]{roy}) is that a projection of a semi-algebraic set is another semialgebraic set, for which the equations can be computed. An algorithmic reformulation that we will extensively rely on is as follows. The \emph{Existential Theory of the Reals} is the following decision problem: we are given a formula of the form $\exists x_1, \ldots , x_n \in \mathbb{R}, \varphi(x_1, \ldots x_n)$ where $\varphi$ is a quantifier-free formula consisting of polynomial equations, polynomial inequalities and Boolean disjunctions and conjunctions. The goal is to decide whether the formula is true. Then the Tarski-Seidenberg theorem shows that the Existential Theory of the Reals is decidable, and the following theorem provides an algorithm to decide it.

\begin{theorem}[{\cite[Theorem~3.12]{roy}}]\label{thm:ETR}
Let $\Phi$ be the formula $\exists x_1, \ldots , x_n \in \mathbb{R}, \varphi(x_1, \ldots x_n)$ where $\varphi$ is a quantifier-free formula consisting of polynomial equations, polynomial inequalities and Boolean disjunctions and conjunctions. Let $s$ be the number of equations and inequalities appearing in $\Phi$, $d$ be an upper bound on their degrees, and $\tau$ be an upper bound on the bitsize of their coefficients. Then there exists an algorithm running in time $(sd)^{O(n)}\poly(\tau)$ that decides the truth of $\Phi$.
\end{theorem}

In particular; one can solve in polynomial time instances of the Existential Theory of the Reals when the number of (in)equations, their degree and the number of variables is constant. This is the case for all the instances of the Existential Theory of the Reals in this paper. In order to ease reading, we will often abuse language in this paper and call an instance of the Existential Theory of the Reals simply a system of polynomial equations.

In our algorithms, we will sometimes want to extract a specific solution to a system of polynomial equations. Such a solution is not provided by Theorem~\ref{thm:ETR}, and this runs into algebraic issues. Indeed, even for a single real polynomial equation, there might be no rational solutions (e.g., for $x^2=2$), and more generally by the Abel-Ruffini theorem shows there is in general no solution in radicals if the degree of the equation is at least five (this is for example the case for $x^5-6x-3=0$~\cite[Example~8.5.5]{cox}).  Nevertheless, we can encode such a solution using \emph{real univariate representations}. This consists of a real single-variable polynomial $f$, an information encoding a single root $t$ of $f$ (its \emph{Thom encoding}~\cite[Definition~2.29]{roy}) and a set of real single-variable polynomials $g_0, \ldots, g_k$. Together, this data represents the point $(\frac{g_1(t)}{g_0(t)}, \ldots ,\frac{g_k(t)}{g_0(t)})$ in $\mathbb{R}^k$. We refer to Basu, Pollack and Roy~\cite[Section~12.4]{roy} for the precise definition and more background. Then Theorem~\ref{thm:ETR} can be strengthened~\cite[Theorem~3.10]{roy} to not only decide if there is a solution, but also compute a real univariate representation of a\footnote{Actually, one can compute a point in each semi-algebraically connected component, but we will not need this stronger fact.} point in the solution set. The corresponding algorithm also has complexity $(sd)^{O(n)}\poly(\tau)$ and the bitsize of the real univariate representation is bounded by $\tau d^{O(n)}$. Throughout this paper, we rely on this algorithm implicitly whenever we say that we \emph{solve} a system of polynomial equations of constant size, and the output is encoded by this linear-size (in $\tau$) real univariate representation. Since this representation amounts essentially to a polynomial equation, it can seamlessly be manipulated, and in particular we can plug it into another system of polynomial equations of constant size, which can then be solved again using the same algorithm.

\paragraph{Matroids.} Our algorithm for Theorem~\ref{thm:additive} relies on enumerating all the possible dependencies within a constant-size subset of vectors of a near-optimal solution. In order to do so, we rely on matroids, which are a combinatorial structure encoding, and generalizing, the dependencies within a family of vectors. We refer to Oxley~\cite{oxley} for an introduction. A \emph{matroid} $M$ is a pair $(E,\mathcal{I})$ where $E$ is a finite set called the \emph{ground set} and $\mathcal{I}$ is a collection of subsets of $E$ called \emph{independent sets} satisfying the following axioms: (i) the empty set is an independent set, (ii) if $I' \subseteq I$ and $I$ is an independent set, $I'$ is an independent set, and (iii) if $I_1$ and $I_2$ are independent sets and $|I_1| < |I_2|$ there exists an element $e$ in $I_2$ such that $I_1 \cup \{e\}$ is an independent set. The \emph{rank} of a matroid is the maximum size of an independent set. It is immediate from the definitions that given a finite set of vectors in $\mathbb{R}^k$, the subsets of vectors which are independent form a matroid. The converse is not true: the matroids that correspond to vectors in $\mathbb{R}^k$ are called \emph{representable over the reals}. We can detect those by encoding dependencies as determinants and using real-algebraic algorithms to solve the corresponding equations as described in the previous paragraph, and in a certain technical sense this is the best algorithm to do so~\cite{kim}. The number of matroids of rank $k$ on a ground set of size $n$ is naturally upper bounded by $2^{n^{O(k)}}$.

The rigid case that we started with in Section~\ref{SS:additive} corresponds to constant-sized sets of vectors $U$ and $W$ in $\mathbb{R}^k$ in which all the subsets of size $k$ form a basis. The matroid obtained from that set is one in which every set of $k$ elements forms an independent set, and in that setting our algorithm is simpler. In order to handle the general case, where sets of $k$ vectors or less might not be independent, our approach requires guessing the entire structure of linear dependencies between elements of $U$ (and $W$). Since the vectors are over the reals, in order to do this guessing, one needs a finite combinatorial abstraction for these dependencies: this is exactly the information that is encoded in the matroids formed by $U$ and $W$.

\paragraph{Extremal graph theory.} The Kovari-S\`os-Turan theorem~\cite{kst} says for any integer $k$, there exists a constant $c_{KST}$ such that any bipartite graph with bipartitions of $n_1$ and $n_2$ vertices with at least $c_{KST} (n_1n_2^{1-1/(k+1)}+n_2)$ edges contains a $K_{k+1,k+1}$ as a subgraph, where the constant $c_{KST}$ is $\Theta(k)$. We will need the following supersaturated version.

\begin{lemma}\label{lem:kst1}
For any $\varepsilon>0$ and $k \in \mathbb{N}$, there exist constants $c'_{KST}=(\varepsilon/k)^{\Theta(k^2)}$ and $n_{KST}=(k/\varepsilon)^{\Theta(k)}$ such that any bipartite graph with bipartitions of $n_1$ and $n_2$ vertices, where $n_1\geq n_2 \geq n_{KST}$ and at least $\varepsilon n_1n_2$ edges contains at least $c'_{KST} n_1^{k+1}n_2^{k+1}$ copies of $K_{k+1,k+1}$ as a subgraph.

\end{lemma}

We include a proof for completeness, it is very similar to the proof of the classical Erd\H{o}s-Simonovits~\cite{erdos} theorem.

\begin{proof}

  Let $n_{KST}$ denote the smallest integer so that $\varepsilon n_{KST} ^{1/(k+1)}\geq 2c_{KST}$, thus $n_{KST}=(k/\varepsilon)^{\Theta(k)}$. For this choice of $n_{KST}$, any $n_1\geq n_2 \geq n_{KST}$ satisfy $(\varepsilon/2) n_1n_2 \geq c_{KST} (n_1n_2^{1-1/(k+1)}+n_2)$.  Let $G$ denote a bipartite graph with at least $\varepsilon n_1n_2$ edges where $n_1\geq n_2 \geq n_{KST}$ denote the sizes of the bipartition.

 We consider the subsets $M$ of $G$ consisting of $n_{KST}$ vertices on each side of the bipartition. Among them, those with at least $(\varepsilon/2) n_{KST}^2$ edges are called \emph{saturated}, and we denote their number by $\eta \binom{n_1}{n_{KST}}\binom{n_2}{n_{KST}}$.

  We double count the number of edges in $G$:

  \[e(G) = \frac{ \sum_{M} e(G[M])}{\binom{n_1-1}{n_{KST}-1}\binom{n_2-1}{n_{KST}-1}} \leq \frac{\eta \binom{n_1}{n_{KST}}\binom{n_2}{n_{KST}} n_{KST}^2+(1-\eta) \binom{n_1}{n_{KST}}\binom{n_2}{n_{KST}}\varepsilon n_{KST}^2}{\binom{n_1-1}{n_{KST}-1}\binom{n_2-1}{n_{KST}-1}}.\]

  By our assumption, $e(G) \geq \varepsilon n_1n_2$, and thus

  \[\varepsilon \leq \eta + (1-\eta) \varepsilon/2.\]

Therefore $\eta \geq \frac{\varepsilon}{2-\varepsilon}$, i.e., the proportion of saturated subgraphs stays bounded away from zero. Each saturated subgraph induced by $M$ contains a $K_{k+1,k+1}$ by the Kovari-S\`os-Turan theorem, and thus $G$ contains (accounting for the multiple counting) at least $\frac{\eta \binom{n_1}{n_{KST}}\binom{n_2}{n_{KST}}}{\binom{n_1-(k+1)}{n_{KST}-(k+1)}\binom{n_2-(k+1)}{n_{KST}-(k+1)}}=\frac{\eta \binom{n_1}{k+1}\binom{n_2}{k+1}}{\binom{n_{KST}}{k+1}\binom{n_{KST}}{k+1}}$ different copies of them. This concludes the proof with $c'_{KST}=(\varepsilon/k)^{\Theta(k^2)}$.
\end{proof}

\section{A polynomial-time additive approximation scheme}

 The main result of this section is the following.

\additive*

\paragraph*{Algorithm and analysis} The algorithm is described in Figure~\ref{A:mainalgo}, where we have used the following notations. We denote by $R$ the set of row indices of $A$ and by $C$ the set of column indices. We think of the matrices $U=(u_{i,j})$ and $W=(w_{i,j})$ as being unknowns, and thus each entry in the matrix $A$ induces an equation $\sum_{\ell=1}^k u_{i,\ell} w_{j,\ell}=A_{i,j}$. Of course, in general, in an optimal solution, not all of these equations will be satisfied. 
\noitemerroroff
\afterpage{%
\thispagestyle{empty}
\begin{figure}[H]
\begin{framed}
\begin{enumerate}[(1)]
\item \label{step:enumeration} For all possible subsets of rows $R_S$ and columns $C_S$, each of size at most $\kappa_1$, for all possible subsets $E \subseteq R_S \times R_C$ and for all rank-$k$ matroids $M_{R_S}$ and $M_{C_S}$ on ground sets $R_S$ and $C_S$,

\begin{enumerate}[(a)]
\item \label{step:solution0} Write a system of polynomial equations $\Xi_1$ where the unknowns are the entries of $U$ and $W$ belonging to $R_S$ and $C_S$ and the constraints are such that:

\begin{itemize}
\item the entries of the matrix $A$ corresponding to $E$ are correct, i.e., for all $(i,j) \in E$, we add the equation $\Xi_{i,j}: \sum_{\ell=1}^k u_{i,\ell} w_{j,\ell}=A_{i,j}$. 

\item the matroid $M_{R_S}$ (resp. $M_{C_S}$) encodes the linear dependencies in $U$ restricted to $R_S$ (resp. in $W$ restricted to $C_S$). This can be encoded by a constant number of equations and inequations involving determinants.

\item (for the restricted version only) the row and columns of $U$ and $W$ must satisfy the projections constraints.
  
\end{itemize}

We denote by $X_1$ the set of solutions to $\Xi_1$.

\item \label{step:solution} Solve this system of polynomial equations over the reals as described in Section~\ref{sec:algeq}. If there is no solution, stop.

\item \label{step:alphabet1} For each row $i$ of $R$ where $i$ is not in $R_S$,

\begin{itemize}
  \item an independent set $I$ of $M_{R_S}$ induces a system of linear equations $\Xi_I$ encoding the fact that $U_i$ belongs to the linear subspace $E_I$ spanned by the vectors of $X_1$ corresponding to the rows indexed by $I$.
  \item an independent set $J$ of $M_{C_S}$ induces a system of linear equations $\Xi_J$ encoding the fact that the entries of the matrix corresponding to $\{i\} \times J$ are correct, i.e., $\Xi_J : \bigcup_{j\in J} \Xi_{i,j}$.
\end{itemize}

The \emph{alphabet} of $i$ is defined as follows: For each subset $\Phi$ of $M_{R_S} \cup M_{C_S}$ of size at most $2k$, take an arbitrary vector that is a solution of the corresponding equations $\bigcup_{I \in \Phi} \Xi_I \cup \bigcup_{J \in \Phi} \Xi_J$ and the projection constraints (in the restricted case), and put it in the alphabet. If there is no such solution, stop.

\item \label{step:alphabet2} Define likewise an alphabet for each column $j$ that is not in $C_S$.

\item \label{step:Max2CSP} We now have define an instance of a Maximization Constraint Satisfaction Problem of arity 2 (Max-2-CSP) where 

\begin{itemize}
\item the variables are the row and columns indices $R \cup C$,
\item the constraint graph is $R \times C$, 
\item for rows and columns not in $R_S$ and $C_S$, the alphabet of each variable is as defined in steps~\ref{step:alphabet1} and~\ref{step:alphabet2}, while for rows and columns in $R_S$ and $C_S$, the alphabet is a single value which is the solution computed in $X_1$,
\item the constraints are whether the two letters agree with the entry of the matrix, i.e. whether $u_i$ and $w_j$ satisfy $\sum_{\ell}u_{i,\ell}w_{j,\ell}=A_{i,j}$, for $u_i$ and $w_j$ letters of the alphabets corresponding respectively to row $i$ and column $j$.

\end{itemize}

\item \label{step:densecsp} We compute an $(\varepsilon/2)$-additive approximation to this dense Max-2-CSP on a constant-size alphabet using standard algorithms (see for example~\cite{yaroslavtsev2014going} and~\cite{manurangsi2017birthday}).

\end{enumerate}
  
\item Output the best solution.

\end{enumerate}
\end{framed}
\caption{Our additive approximation scheme for $\ell_0$-Low-Rank Approximation.}
\label{A:mainalgo}
\end{figure}
}

The proof of Theorem~\ref{thm:additive} follows from Proposition~\ref{P:supercore}. It is quite a bit stronger than what is actually needed for Theorem~\ref{thm:additive}, as this stronger version will be required for the proof of Theorem~\ref{thm:main}. We actually solve a more constrained problem, \textsc{Restricted}-$\ell_0$-\textsc{Low Rank Approximation}, where we are additionally given a pair of \emph{projection constraints}, that is, matrices $p_R$ in $\mathbb{R}^{t_R \times k}$ and $p_C$ in $\mathbb{R}^{t_C \times k}$ as well as real vectors $(a_i)_{i\in R}$ in $\mathbb{R}^{t_R}$ and $(b_j)_{j \in C}$ in $\mathbb{R}^{t_c}$, and we require that the matrices $U$ and $W$ also satisfy $p_R(u_i)=a_i$ and $p_C(v_j)=b_j$ for all $i$ and $j$.

The algorithm is parameterized by a large constant $\kappa_1=\kappa_1(k,\varepsilon)$ which, as we will see later, can be taken to be $(k/\varepsilon)^{\Theta(k^3)}$. We define a \emph{supercore} $\Sigma=(R_S,C_S,G,M_{R_S},M_{C_S})$ as being the data enumerated in step~\ref{step:enumeration} of the algorithm: subsets $R_S$ and $C_S$ of rows and columns of the same size, a bipartite graph on these subsets $G=(R_S \cup C_S,E)$ and a pair of rank-$k$ matroids $M_{R_S}$ and $M_{C_S}$ on $R_S$ and $C_S$. The size of a supercore is the size of $R_S$ and $C_S$.

\begin{proposition}\label{P:supercore}
 Let $\varepsilon>0$ and $k \in \mathbb{N}$ be constants. There exists a constant $\kappa_1=(k/\varepsilon)^{\poly(k)}$ such that for any $n_R \times n_C$ matrix $A$, any pair of projection constraints, and any $n_1 \times n_2$ submatrix $A' \subseteq A$, there exists an $(\varepsilon/2)$-additive approximation to \textsc{Restricted} $\ell_0$-\textsc{Low Rank Approximation} on $A'$, that we call Near-OPT such that one of the supercores $\Sigma=(R_S,C_S,G,M_R,M_C)$ of size at most $\kappa_1$ satisfies:
  \begin{enumerate}[(1)]
\item The rows and columns of Near-OPT indexed by $R_S$ and $C_S$ (restricted to $A'$) match those of the solution computed in Step~\ref{step:solution}.
\item The other rows and columns of Near-OPT are contained in the alphabets computed in steps~\ref{step:alphabet1} and~\ref{step:alphabet2}.    
\end{enumerate}
\end{proposition}
\noitemerroron

This proposition immediately implies Theorem~\ref{thm:additive} by using the full matrix $A$ for $A'$ and enforcing no projection constraints on the rows and the columns: it shows that the Max-2-CSP that we define in Step~\ref{step:Max2CSP} provides an $\varepsilon/2$-additive approximation to $\ell_0$-\textsc{Low Rank Approximation}, and thus solving this Max-2-CSP with an $\varepsilon/2$-additive approximation will yield the desired $\varepsilon$-additive approximation. The complexity of the algorithm is dominated by Step~\ref{step:enumeration}, where the algorithm enumerates all subsets of the rows and columns of size at most $\kappa_1=(k/\varepsilon)^{\Theta(k^3)}$. The other heavy computational steps are solving the system of polynomial equations in Steps:~\ref{step:solution},~\ref{step:alphabet1} and~\ref{step:alphabet2}: the number of equations is always upper bounded by ${\kappa_1^{\poly(k)}}$, and thus applying Theorem~\ref{thm:ETR} , we stay within the allowed timebound. Finally, solving additively the Max-2-CSP instance in Step~\ref{step:densecsp} can be done in time $q^{O(\log q/\eps^2)}+\poly(n)$, where $q$ is an upper bound on the size of the alphabet, which we can take to be ${\kappa_1^{\poly(k)}}$ (see~\cite{yaroslavtsev2014going} and~\cite[Footnote 2]{manurangsi2017birthday}).

We now prove Proposition~\ref{P:supercore}.

\begin{proof}

Throughout this proof, we reason on the submatrix $A'$. Without loss of generality we can assume that $n_1 \geq n_2$. We denote a solution to \textsc{Restricted} $\ell_0$-\textsc{Low Rank Approximation} in this submatrix by $Sol=(U,W)$, and we denote by $G$ the bipartite graph $G=(R' \cup C',E)$, where $R'$ and $C'$ denote the indices of rows and columns of $A'$ and $(i,j) \in E$ if the $(i,j)$ entry of the matrix $A'$ agrees with $(UW^T)_{i,j}$. Initially, $Sol$ will  be an optimal solution, which we will then modify to a near-optimal solution, i.e., a $\varepsilon/2$-additive approximation to the optimal solution.%, and we sometimes write $G$ for $G_{Sol}$ to ease notation. %We use abundantly the notation $O(1)$, where the notation hides dependencies on $k$, and $O_{\varepsilon}(1)$, where the notation also hides dependencies on $\varepsilon$. %Note that we can assume that $Sol$ has at least $\varepsilon/2 \  n^2$ edges, as otherwise any solution is an $\varepsilon/2 \  n^2$-additive approximation, and thus the lemma is trivial.

The first step of the proof is to show that there exists a near-optimal solution which has a nice structure for our problem. Here, ``nice" means that such a solution is parameterized by constantly many pieces, each of which can be fully determined from a subset of constant size (its \emph{core}). %We want to show that there exists a near-optimal solution in which the

In order to do so, we start from an 
optimal solution, corresponding to a graph $G$, and define a
family of pieces as follows. We consider a set
$S := R_S \cup C_S \subset R' \cup C'$ 
of $\kappa_2$ rows and $\kappa_2$ columns, where $\kappa_2=\kappa_2(k,\varepsilon)$ is a constant to be fixed later. The optimal solution induces a pair of rank-$k$ matroids $M_R$ and $M_C$, which restrict to submatroids $M_{R_S}$ and $M_{C_S}$ on the sets $R_S$ and $C_S$, and, for each independent set $I \in M_R$ (respectively $J \in M_C$), there is a corresponding subspace $E_I$ (respectively $E_J$).

For each pair of independent sets $I,J$ in $M_{R_S} \times M_{C_S}$,
we define a piece $P_{I,J}$ as follows.  
If the complete bipartite graph on $I\times J$ is included in $G$, we create a piece with $I\cup J$ as the core and $I \cup I_2 \cup J \cup J_2$ as the vertex set, where $J_2$ denotes the set of vertices adjacent to all of $I$ in $G$ and belonging to the subspace spanned by $J$, and $I_2$ denotes the set of vertices adjacent to all of $J$ in $G$ and belonging to the subspace spanned by $I$. The graph of the piece is the subgraph of $G$ on the vertex set. The matroids on this piece are the submatroids induced by $M_R$ and $M_C$. Note that pieces will in general overlap.%Note that this construction is not symmetric, a similar construction is done symmetrically to define pieces $P_{J,I}$ for each pair of independent sets $J,I$ in $M_{C_S} \times M_{R_S}$.

\begin{lemma}\label{lem:kst}
If $\kappa_2=(k/\varepsilon)^{\Omega(k^2)}$, there exists a set $S$ such that all but at most $(\varepsilon/2)n_1n_2$ edges of $G$ belong to at least one piece $P_{I,J}$.
\end{lemma}

\begin{proof}

  We distinguish two cases for the proof, depending on how $n_2$ compares to $n_{KST}$, the constant of Lemma~\ref{lem:kst1}.

\paragraph{First case: $n_2 \geq n_{KST}$.}
In this case, the proof relies on the probabilistic method: the set $S$ is taken by sampling uniformly at random a set $R_S$ of $\kappa_2$ rows and a set $C_S$ of $\kappa_2$ columns.  Let $\kappa_3$ be a constant depending on $\varepsilon$ and $k$ to be fixed later. We say that an edge $e=(i, j)$ is \emph{efficient} if there exists at least $\kappa_3 n_1^kn_2^k$ copies of $K_{k+1,k+1}$ in $G$ containing $e$ such that for each such copy, the space spanned by (the vectors corresponding to) the other $k$ vertices in $R$ contains $i$, and the space spanned by (the vectors corresponding to) the other $k$ vertices in $C$ contains $j$. The main claim that we prove is:

\begin{claim}
All but at most $(\varepsilon/4)n_1n_2$ edges are efficient.
\end{claim}

  \begin{proof}[Proof of the claim.] We say that a $K_{k+1,k+1}$ is good for one of its edges $(i, j)$ if $i$ and $j$ belong to the span of the other respective $k$ vertices. Let us assume that the claim is wrong. Then there are at least $(\varepsilon/4)n_1n_2$ edges $(i,j)$ which do not belong to at least $\kappa_3 n_1^kn_2^k$ copies of $K_{k+1,k+1}$ which are good for them. We remove all the other edges from the graph, call the resulting graph $G'$. Then we apply the (supersaturated) Kovari-Sos-Turan theorem of Lemma~\ref{lem:kst1} on this $G'$. The assumptions hold since by the assumption of the first case, $n_1 \geq n_2 \geq n_{KST}$ and there are at least $(\varepsilon/4) n_1n_2$ edges. It implies that that there exists $c'_{KST}=(\eps/k)^{\Theta(k^2)}$ such that there are at least $c'_{KST}n_1^{k+1}n_2^{k+1}$ copies of $K_{k+1,k+1}$ in $G'$. Now, observe that each copy of $K_{k+1,k+1}$ is good for at least one of its edges. Thus, by double counting, there must be one edge contained in $c'_{KST} n_1^{k}n_2^{k}$ copies of $K_{k+1,k+1}$ which are good for it. 
  This is a contradiction for $\kappa_3\leq c'_{KST}=(\varepsilon/k)^{\Theta(k^2)}$.
    \end{proof}
  
  We now prove Lemma~\ref{lem:kst}. We first discard the inefficient edges. Then we claim that with probability more than $1-\varepsilon/4$, an efficient edge belongs to a piece. In order to prove that, we partition the sample set $S$ into $\lfloor \kappa_2/(2k)\rfloor$ disjoint subsets of $k$ rows and $k$ columns. With probability at least $\kappa_3$, such a subset will contain the $2k$ other vertices of one of the $K_{k+1,k+1}$ defining an efficient edge. So if $(1-\kappa_3)^{\kappa_2/(2k)} \leq \varepsilon/4$, which happens if the constant $\kappa_2$ sufficiently large (using the value of $\kappa_3$ from the claim, we can take $\kappa_2=(k/\varepsilon)^{\Theta(k^2)}$), with probability more than $1-\varepsilon/4$ the sample will contain the $2k$ other vertices of one of the $K_{k+1,k+1}$ defining an efficient edge.

  We consider subsets $I$ and $J$ of these $2k$ vertices so that $I$ and $J$ are independent sets, and $I$ is dependent with $i$ and $J$ is dependent with $j$. Then the edge $e$ will belong to the piece $P_{I,J}$ since $I$ and $J$ induce a complete bipartite graph, and thus form a core, and by definition of efficiency, $i$ and $j$ are adjacent to all the vertices of this core. By linearity of expectation, the expected number of edges for which this fails is less than $(\varepsilon/4)n_1n_2$. Thus with nonzero probability, the set $S$ has the required properties. This concludes the proof.

\paragraph{Second case: $n_2 \leq n_{KST}$.}
  In that case, we even have the stronger result that there exists a choice of $S$ such that every edge belongs to a piece. Indeed, for each column $j$, denote by $I(j)$ the set $\{i \in R' \mid (i, j) \in G\}$, and by $B(j)$ a subset of $I(j)$ so that the rows indexed by $B(j)$ form a basis of the vector space spanned by the rows indexed by $I(j)$ in the optimal solution. Then we consider the set of rows obtained by taking the union of all the sets $B(j)$. We take the set $S$ to consist of the union of these rows and the entire set of columns, which we can do if $\kappa_2 \geq kn_{KST}$ (recall that $n_{KST}=(k/\varepsilon)^{\Theta(k)}$). Now, each edge $(i, j)$ of $G$ is contained in the piece $P_{B(j),\{j\}}$. Indeed this piece exists since by definition the graph induced by $B(j)$ and $\{j\}$ in the optimal solution is the complete bipartite graph. Furthermore, $i$ belongs to the space $B(j)$ and is adjacent to $j$, therefore it belongs to this piece.\end{proof}

From now on, we consider that the set $S$ satisfies Lemma~\ref{lem:kst}. At the cost of modifying the solution by $(\eps/2) n_1n_2$, we can neglect the edges not covered by Lemma~\ref{lem:kst}, and therefore assume that all the edges (and in particular all the vertices) are contained in some piece, which we do from now on. %(Formally, instead of working with the optimal solution, we work with a near-optimal solution, note that this does not change the fact that the $P_{I,J}$ are pieces).

Now, we aim at controlling the interactions between different pieces $P_{I,J}$. This is done by defining the following \emph{auxiliary cores}, which is an additional set of subgraphs of $G$. Let $\mathcal{P}$ denote a set of at most $k$ pairs of independent sets $(I,J)$ in $M_{R_S} \times M_{C_S}$. If the intersection of all the rows in $P_{I,J}$ for $(I,J)\in \mathcal{P}$ is non-empty, we let $I'$ denote a subset of rows in this intersection for which the vectors are maximally independent. Then we define an auxiliary core whose vertex set is $I' \cup \bigcup_{\exists I,(I,J) \in \mathcal{P}} J$ and whose graph is the complete bipartite graph on the vertex set. Note that by construction of a piece $P_{I,J}$, each row of such a piece is adjacent to all the columns in $J$. Therefore, the auxiliary core is indeed a subgraph of $G$. We also define auxiliary cores symmetrically with the roles of rows and columns inverted.

The supercore is defined as the union of the cores of all the pieces $P_{I,J}$
and all the auxiliary cores for all sets $\mathcal{P}$ of at most $k$ pairs of independent sets. The corresponding matroid is the one induced from the near-optimal solution. By construction, it has size at most some constant $\kappa_1$ which we can take to be $\kappa_2^{\Theta(k)}=(\frac{k}{\varepsilon})^{\Theta(k^3)}$, and we are now ready to prove Proposition~\ref{P:supercore}. We consider a near-optimal solution $Sol=(U,W)$ where pieces cover all the edges, as provided by Lemma~\ref{lem:kst}, and denote by $G$ the graph of covered edges. We then consider an arbitrary solution $Sys=(U_{sys},W_{sys})$ to the system of equations in Step~\ref{step:solution0}, where the correct supercore has been guessed. A solution always exists since $Sol$ is such a solution. We will show how to extend $Sys$ to a solution of the whole set $R' \cup C'$ such that the set of edges in $G$ is satisfied.

We first consider an edge $(i, j)$ that is not included in the supercore, and such that $i$ (respectively $j$) does not belong to $R_S$ (respectively $C_S$). In $Sol$, this edge is included in a collection of pieces $P_{I,J}$. Each of these pieces puts two types of constraints on the value of the row $i$: (i) it should belong to the subspace spanned by the vectors of $I$ in $Sys$ and $(ii)$ it should satisfy the linear equations induced by the core of $P_{I,J}$, where the values of the columns in $J$ are fixed by $W_{sys}$. The pieces put a symmetric set of constraints on the column $j$. We claim that if these constraints are satisfied in a solution that extends $Sys$, all the edges of $G$ are automatically satisfied in that solution. This will follow from this easy linear algebraic lemma.

\begin{lemma}\label{lem:linear}
If $M_1$ and $M_2$ are two matrices and $a$ and $b$ are two real vectors, then for any two vectors $u$ and $w$ such that $M_1u=a$, $M_2w=b$, $u$ is a linear combination of the columns of $M_2$ and $w$ is a linear combination of the columns of $M_1$, the value of $\langle u,w\rangle$ is uniquely determined.
\end{lemma}

\begin{proof}
Let $u_1,w_1$ and $u_2,w_2$ be two pairs of vectors satisfying the conditions of the lemma. Then since $M_2w_1=M_2w_2=b$, $w_2-w_1$ is in the kernel of $M_2$. Since $u_1$ is a linear combination of the columns of $M_2$, there exists a vector $x$ such that $u_1=M_2^Tx$, and then $\langle u_1,w_2\rangle=\langle u_1,(w_1+w_2-w_1)\rangle=\langle M_2^Tx,w_1\rangle +\langle M_2^Tx,w_2-w_1\rangle=\langle M_2^Tx,w_1\rangle=\langle u_1,w_1\rangle$. Likewise, writing $w_2=M_1^Ty$, $\langle u_2,w_2\rangle =\langle u_1+u_2-u_1,M_1^Ty\rangle=\langle u_1,w_2\rangle$, and thus $\langle u_1,w_1\rangle=\langle u_2,w_2\rangle$.
\end{proof}

For each edge $(i, j)$ in a piece $P_{I,J}$, if we write $M_1u_i^T=a$ and $M_2w_j=b$ for the linear equations of type (ii) induced respectively by the columns and the rows of the core on the vectors $u_i$ and $w_j$, then the conditions (i) directly imply that $u_i^T$ is a linear combination of the columns of $M_2$ and $w_j$ is a linear combination of the columns of $M_1$. Therefore, by Lemma~\ref{lem:linear}, the value of $\langle u_i,w_j\rangle$, which determines whether the edge $(i, j)$ is satisfied, is unique, and in particular is equal to its value in $Sol$, where it is satisfied by definition of an edge of $G$.

So in order to satisfy all the edges of $G$, it suffices to extend $Sys$ to a solution of the whole set $R' \cup C'$ such that

\begin{enumerate}[label=(\roman*)]
\item the linear dependencies between a row and the rows of the cores are the same as in $Sol$,
\item the linear dependencies between a column and the columns of the cores are the same as in $Sol$,
\item for any row $i$, all the edges to the core columns of pieces that $i$ belongs to are satisfied, 
\item for any column $j$, all the edges to core rows of pieces that $j$ belongs to are satisfied,
\item in the restricted case, the additional projection constraints are satisfied.
\end{enumerate}

We claim that this can always be done. Let $i$ be a row that belongs to a set of pieces $\{P_{I,J}\}$, and denote by $\mathcal{P}$ the corresponding set of pairs of independent sets. Since the constraints induced by each piece $P_{I,J}$ are linear and the row $i$ belongs to $\mathbb{R}^k$, there exists a subset $\mathcal{P}'$ of $\mathcal{P}$ of size at most $k$ inducing exactly the exact same constraints as $\mathcal{P}$ on the row $i$.

We consider the system of equations in Step~\ref{step:alphabet1} obtained by taking for $\Phi$ the set of independent sets involved in $\mathcal{P}'$. We directly have that any solution to these system of equations satisfy the constraints above corresponding to rows. So there remains to show that this system of equations has a solution. This is immediate for a row $i$ that belongs to the core of one of the pieces or to the auxiliary core corresponding to $\mathcal{P}'$.

Otherwise, we denote by $I'$ the set of rows of the auxiliary core corresponding to $\mathcal{P}'$. We denote by $u^*_i$ and $w^*_j$ the vectors of $Sol$, and by $u_i$ and $w_j$ the row and vectors obtained as a solution of the system of equations in Step~\ref{step:solution}. By definition, in the solution $Sol$ the row $i$ belongs to the space spanned by the rows $I'$, therefore we can write $u^{*}_i=\sum_{\ell \in I'} \alpha_{\ell} u^{*}_\ell$, for some family of real numbers $\alpha_{\ell}$. We want to show that there exists $u_i$ for the row $i$ with the following constraints: (i) the inner products induced by edges with the columns of the cores of the pieces $P_{I,J}$ that $i$ belongs to are satisfied, (ii) $u_i$ belongs to the space spanned by the row vectors in $Sys$ indexed by $I'$ and (iii) in the restricted case, $u_i$ satisfies the projection constraint $p_R(u_i)=b_i$. We consider the vector $u_i:= \sum_{\ell \in I'} \alpha_{\ell} u_\ell$ and claim that it satisfies all three constraints. It trivially satisfies (ii). For any column $j$ in a piece $P_{I,J}$ in $\mathcal{P}'$, we have

\begin{align*}\langle u_i,w_j\rangle =& \langle \sum_{\ell \in I'} \alpha_\ell u_\ell,w_j \rangle 
  =\sum_{\ell \in I'} \alpha_\ell \langle u_\ell,w_j\rangle 
  =\sum_{\ell \in I'} \alpha_\ell \langle u_\ell,w_j\rangle
  =\sum_{\ell \in I'} \alpha_\ell A_{\ell,j}\\
  =&\sum_{\ell \in I'} \alpha_\ell \langle u^{*}_\ell,w_j^{*}\rangle
  = \langle u^{*}_i,w_j^{*}\rangle 
  = A_{i,j}.
\end{align*}

By definition of $\mathcal{P}'$, the constraints induced by the pieces in $\mathcal{P} \setminus \mathcal{P'}$ are also satisfied. For condition (iii), we have similarly

\begin{align*}
  p_R(u_i)=\sum_{\ell \in I'} \alpha_\ell p_R(u_\ell)
  =\sum_{\ell \in I'} \alpha_\ell b_\ell
  =\sum_{\ell \in I'} \alpha_\ell p_R(u^*_\ell)
  =p_R(u^*_i)
  =&b_i
\end{align*}

Symmetrically, we can always find a solution to all the constraints for any column $j$. Therefore, we can extend $Sys$ to a solution that is at least as good as $Sol$, and is thus at least a $\varepsilon/2$-additive approximation to \textsc{Restricted} $\ell_0$-\textsc{Low Rank Approximation} on $A'$. This concludes the proof of the proposition. 
\end{proof}

\input{LRA-min-ptas-Final.tex}

\input{LRA-hardness.tex}

  \bibliographystyle{alpha}
  \bibliography{biblio}

\end{document}

%% file: techniques-min.tex
Based on the additive approximation scheme, we introduce high-level ideas between our multiplicative approximation scheme for Theorem~\ref{thm:main}.
Like the additive approximation scheme, our algorithm also works by reducing the alphabet size to a constant. In order to do so, our new framework here partitions the set of entries $[n_R] \times [n_C]$ into rectangular {\em blocks} (there are at most $k \times k$ of them) and handle them separately in the following natural ways: if a block $B$ has 
\begin{itemize}
\item $|B| \gg OPT$ (called {\em clean}): Techniques for constant-size alphabets almost suffice, as random entries from $B$ are correct in the optimal solution and reveal useful information about it. 
\item $|B| \ll OPT$  (called {\em dirty}): We can ignore $B$ as it will not contribute much to the objective. 
\item $|B| \approx OPT$  (called {\em half-clean}): Use the additive PTAS, because an additive approximation is also a multiplicative approximation in this case. 
\end{itemize}

One technical and conceptual challenge is that the algorithm will never be able to learn where the blocks are, but our algorithm still manages to handle them using careful definitions of the blocks (only in the analysis) and the additional features of the additive PTAS. We shall explain the ideas in more detail below. 

\paragraph{Basic strategy.} 
Recall that given an instance of \lra $A \in \R^{n_R \times n_C}$ with an optimal solution $UW^T$, 
we view this as a CSP where there is a variable for each row and column, 
and the goal is to choose a value $u_i$ from the alphabet $\Sigma = \R^k$ for each row $i$ and $w_j \in \Sigma$ for each column $j$ to satisfy the constraints given by $A$. 
It is natural to review previous approaches for dense Min-CSPs~\cite{karpinski2009linear, voting} and Low Rank Approximation on Finite Domains~\cite{ban2019ptas}.
With an oversimplification that ignores important technical details, their main ideas, when the alphabet set is general $\Sigma$, can be roughly summarized as: 
\begin{enumerate}
\item Sample a constant number of column indices $s_1, \dots, s_t \in [n_C]$. 
\item Guess the {\em value} of each $s_i$ in the optimal solution; i.e., guess $w_{s_i} \in \Sigma$.
\item Based on $w_{s_1}, \dots, w_{s_t}$, {\em greedily choose} $u_p$ for each $p \in [n_R]$; i.e., choose $u_p$ that makes the least number of errors with $w_{s_1}, \dots, w_{s_t}$.
\item Given the value of every $u_p$, greedily choose $w_q$ for every $q \in [n_C]$. 
\end{enumerate}

Having $\Sigma = \R^k$ presents a challenge in almost every step. 
For us, the biggest challenge is Step 2, where we cannot guess the values of the sampled columns in the optimal solution via exhaustive enumeration. 
Therefore, our overall goal is to {\em reduce the alphabet set from $\R^k$ to constant-size sets};
formally, our algorithm will construct the alphabet set $\Sigma_p$ for each row and column $p$ with $|\Sigma_p| \leq O_{k, \eps}(1)$
so that there exists a near-optimal solution where each row and column draws a value from their given alphabet sets.
(Actually, the algorithm creates polynomially many such instances with the guarantee that one of them contains a near-optimal solution using their alphabets.)

Another smaller challenge related to Step 3 is that even given the correct values for $w_{s_1}, \dots, w_{s_t}$, possibly none of $u_i$'s can be determined.
For example, if most of the rows and columns belong to a proper subspace $T \subsetneq \R^k$ and all $w_{s_1}, \dots, w_{s_t}$ are in $T$, then at best the algorithm can determine $u_p$'s projection to $T$, but not its exact position in $\R^k$ (while most of the errors made by the optimal solution might come from the few $u_p$'s and $w_q$'s outside $T$).
Inspired by the previous approaches, we handle this issue by dividing $\R^k$ (and the set of rows and columns) into {\em layers} and obtain uniform samples from each layer. 
The algorithm will not know the layers, so {\em sampling} for the rest of the subsection is just needed to show the existence of a good seed set. The algorithm will enumerate all possible seed sets of certain size.

Our {\em column layers} are sets $J_1, \dots, J_{\ell_C}$ with $\ell_C \leq k$ that partition $[n_C]$ with associated subspaces 
$\emptyset = T_{J, 0} \subseteq T_{J, 1} \subseteq \dots \subseteq T_{J, {\ell_C}} = \R^k$ such that for any $j \in [\ell_C]$, $\{ w_q : q \in J_j \} \subseteq T_{J, j}$.
We require that the layer sizes are decreasing quickly (e.g., $|J_j| \leq \alpha |J_{j-1}|$) and crucially, each layer is {\em full}; 
for any $j \in [\ell_C]$, no subspace $T'$ with $T_{J, {j-1}} \subseteq T' \subsetneq T_{J, j}$ contains more than a $(1-\beta)$ fraction of $J_j$ with some constants $\alpha, \beta \in (0, 1)$. Even though in the actual algorithm they are both set to be constants depending only on $k$, for simplicity of this overview, let us make the key simplifying assumption that $\alpha = o_n(1)$ while $\beta$ is still a constant. (This will avoid the notion of superlayers and hyperlayers in Section~\ref{sec:min}.)

Once we obtain samples $\{ s_{j, 1}, \dots, s_{j, t} \}$ from each layer $J_j$ and guess their values $\{ w_{s_{j, q}} \}_{q \in [t]}$
in a near-optimal solution, one can show that the standard algorithm, choosing greedily $u_p$ for every $p \in [n_R]$ and choosing greedily $w_q$ for every $q \in [n_C]$ guarantees a good solution. (See Phase 4 of Section~\ref{sec:min} for details.) 
Therefore, once the the alphabet size for columns becomes a constant, the algorithm can obtain samples, guess the values of the samples, and perform the greedy decisions to obtain a PTAS.

Now we describe our main alphabet-reduction algorithm to construct a constant-size alphabet set for each row and column. Note that in the beginning, the algorithm has no information about the initial optimal solution $U$ and $W$. 
While describing the algorithm, we will also transform $U$ and $W$ such that 
(1) the transformed solution is still near-optimal, and (2) the algorithm acquires more information about them as it proceeds. 
Just like Section~\ref{sec:min}, we present this algorithm in three phases.

\paragraph{Phase 1: Obtaining initial samples.} 
Our alphabet-reduction algorithm also begins with sampling. 
As well as the column layers, construct the {\em row layers} $I_1, \dots, I_{\ell_R}$ with the subspaces $T_{I, 1}, \dots, T_{I, \ell_R}$ for some $\ell_R \leq k$. 
Call $B_{i,j} := I_i \times J_j$ the $(i, j)$th block. Then we have $\ell_R \times \ell_C$ {\em blocks}.
By guessing, we can assume that the algorithm knows all block sizes and $OPT$, the number of errors that the optimal solution makes. (There are $n^{O(k)}$ possibilities). 

Then we compare the size of each block $|B_{i,j}|$ to $OPT$. 
If $|B_{i,j}| \ll OPT$, we call it {\em dirty}; we can afford to make errors in the entire block, so we can safely ignore it. 
Otherwise, if $|B_{i,j}| \gg OPT$, we call it {\em clean}; we can get a lot of information of this block by samples, because when we uniformly sample rows from $I_i$ and columns from $J_j$, most of the entries between them are correct; the entries of the input matrix $A$ are the correct inner product values between optimal vectors. We call all other blocks {\em half-clean}. 
As these definitions only depend on the sizes of the blocks, we have a natural monotonicity property: for instance, if $B_{i, j}$ is clean then $B_{i-1, j}$ is as well, and if $B_{i, j}$ is dirty, $B_{i+1, j}$ is dirty too. Another crucial consequence of this definition (and our key simplifying assumption) is that each row and column belongs to at most one half-clean block. This will be important when the algorithm applies the additive PTAS in Phase 3. See Figure~\ref{F:blocks_tec} for an example. 

\begin{figure}
    \centering
	\includegraphics[width=12cm]{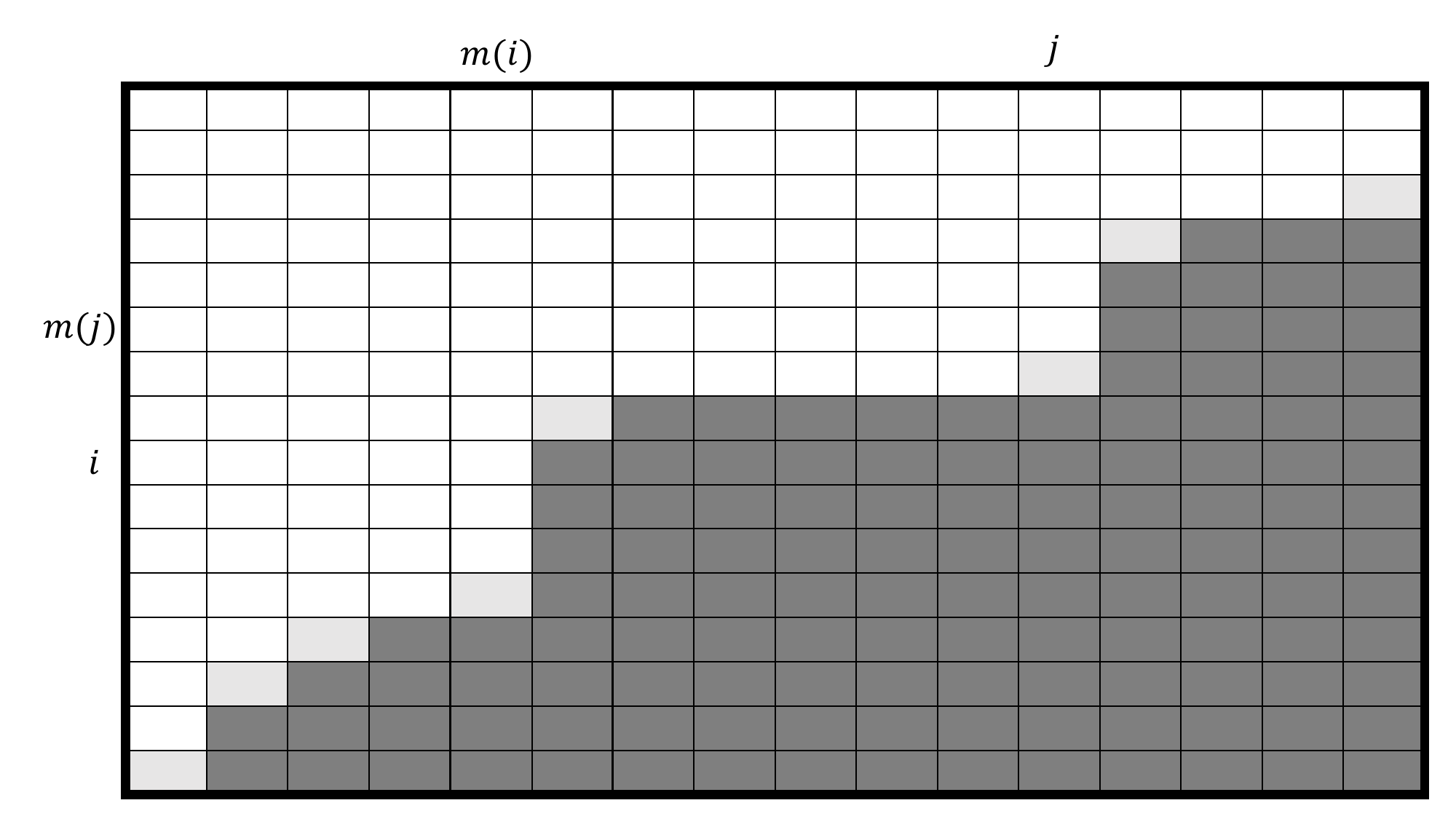}
    \caption{There are $16 \times 16$ blocks. Empty cells are clean, half-shaded cells are half-clean, and full-shaded cells are dirty.}
    \label{F:blocks_tec}
\end{figure}

The algorithm obtains samples from each row and column layer; call them $\{ r_{i, p} \}_{i \in [\ell_R], p \in [t]}$ and $\{ s_{j, q} \}_{j \in [\ell_C], q \in [t]}$ for some constant $t$. Since each $T_{I, i}$ is full, the row samples will contain a basis of $\R^k$, 
so by applying an appropriate transformation $U \leftarrow U C, W \leftarrow W (C^{-1})^T$ for some invertible $C \in \R^{k \times k}$ (only in the analysis), 
the algorithm knows $T_{I, i}$ for every $i \in [\ell_R]$. 
For a clean block $B_{i, j}$, using the correct entries between the sampled rows and columns, the algorithm can even recover $T_{i, j}$, which is the projection of $T_{C, j}$ to $T_{I, i}$ and losslessly captures the interaction between the vectors in $B_{i, j}$.

\paragraph{Phase 2: Handling clean blocks.}
One (non-)feature of our alphabet-reduction algorithm is that it will never determine whether a particular row or column belongs to a certain layer. 
Instead, for each row-layer pair $(p, i) \in [n_R] \times [\ell_R]$, 
the algorithm will construct a set of vectors $\Sigma_{p, i}$ that contains the correct vector $u_p$ if $p$ indeed belongs to $I_i$ in the current near-optimal solution $(U, W)$. 
In this phase, we begin this process from clean blocks, where each $(p, i)$ pair chooses only one vector $u_{p, i}$ {\em inside $T_{i, m(i)}$ instead of $\R^k$}, where $m(i) \in [\ell_C]$ is the largest index $j$ where $B_{i, j}$ is clean. In particular, when $p \in I_i$ in the current near-optimal solution, $u_{p, i}$ is exactly the projection of $u_p$ to $T_{i, m(i)}$, denoted by $u_p|_{T_{i,m(i)}}$.

More concretely, for each row $p \in [n_R]$ and $i \in [\ell_R]$, we let the column samples $\{ s_{j, q} \}_{j \in [m(i)], q \in [t]}$ vote for the projection of $u_p$ to $T_{i, m(i)}$ and call the winner $u_{p, i}$. Formally,
\[
u_{p, i} = \argmin_{u \in T_{i, m(i)}} \sum_{j = 1}^{m(i)} | \{ q \in [t] : A_{p, s_{j, q}} \neq \langle u, w_{s_{j, q}}|_{T_{i, j}} \rangle \} | \cdot (|J_{j}| / t)
\]
where ties are broken arbitrarily. 
The goal is to ensure that $u_{p, i} = u_p|_{T_{i, m(i)}}$ if $p \in I_i$. 
Of course, this cannot happen always, but we will conduct the following transformation that forces it. 

\begin{itemize}
\item For every $p \in [n_R]$, let $i \in [\ell_R]$ such that $p \in I_i$. 
\begin{itemize}
\item If $u_{p,i}$ is indeed the projection of $u_p$ to $T_{i, m(i)}$, then do not change anything. 
\item Otherwise, say $p$ is {\em mistaken} and let $u_p \leftarrow u_{p, i}$. By doing this, we (conservatively) make every entry $A_{p, q}$ with $q \in J_j$ and $j > m(i)$ incorrect.
\end{itemize}
\end{itemize}

Our main technical lemma (Lemma~\ref{lem:min}) shows that this transformation of $U$ ensures that the solution pair $(U, W)$ is still near-optimal. 
In particular, it shows that (1) the chosen $u_{p, i}$ will be again nearly optimal in the clean blocks $B_{i, 1}, \dots, B_{i, m(i)}$, and 
(2) the probability that $p$ is mistaken is small so that the additional error in half-clean or dirty blocks due to a mistake will be small in expectation. 
We do the almost same for columns to compute column vectors $w_{q, j} \in T_{m(j), j}$ for each $q \in [n_C]$ and $j \in [\ell_C]$.

\paragraph{Phase 3:  Handling half-clean blocks.} 

Finally, the algorithm constructs the alphabet set that will contain a good solution for half-clean blocks as well. For one row-layer pair $(p, i)$, we construct the set of vectors $\Sigma_{p, i}$ such that (1) for any $u \in \Sigma_{p, i}$, the projection of $u$ to $T_{i, m(i)}$ is equal to $u_{p, i}$ constructed in the previous phase, and (2) if $p$ indeed belongs to $I_i$ in the current near-optimal solution,
then $\Sigma_{p, i}$ contains $u_p$, the correct vector of $p$ in the current near-optimal solution. 

The main idea here is to apply the additive PTAS to every half-clean block $B_{i, j}$. By definition, $|B_{i,j}| = \Theta(OPT)$, so an additive $\eps$-approximation in $B_{i, j}$ will lead to an overall multiplicative approximation. 
But the crucial bottleneck is that we will never know where $B_{i, j}$ is! As previously mentioned, our alphabet-reduction algorithm will never determine $p \in I_i$ for any $p \in [n_R]$ and $i \in [\ell_R]$. 
What we do know is $u_{p, i}$, which is the correct projection of $p$'s near-optimal vector $u_p$ to $T_{i, j}$, if $p$ indeed belongs to $I_i$. 

We resolve this issue by, for every half-clean block $B_{i, j}$, running the additive PTAS algorithm for the entire matrix $A$ {\em pretending} that every row belongs to $I_i$ and every column belongs to $J_j$. Though there are exponentially many candidates for $B_{i,j}$ inside $A$, the structure of our additive PTAS guarantees that any submatrix of $A$ corresponding to a block $I' \times J'$ with $I' \subseteq [n_R], J' \subseteq [n_C]$ will admit a constant-size subset of $I^+ \subseteq I'$ and $J^+ \subseteq J'$ that suggest a set of vectors for everyone in $I \cup J$ containing their correct vectors in the near-optimal solution. Then, even without knowing actual $I_i$ and $J_j$, the algorithm can try all constant-size subsets $I^+$ and $J^+$ that suggest a constant-size alphabet for every $p \in [n_R]$ and $q \in [n_C]$! Of course, if $i \notin I_i$, then this suggestion does not have any guarantee, but we do know that if $i \in I_i$, this suggestion will contain the correct vector. 
Also note that this strategy depends on the fact that each row or column belongs to at most one half-clean block; otherwise, one row would have received more than one ``correct suggestions'' where each correct suggestion yields a good solution for only one half-clean block.

Therefore, we run the additive PTAS for each half-clean hyperblock, 
and for each choice of $(I^+, J^+)$'s, 
we have an instance of \lraa where each row or column $p$ gets a constant-size alphabet set $\Sigma_{p} = \cup_{i} \Sigma_{p, i}$, with the guarantee that, for at least one choice of $(I^+, J^+)$'s, there exists a near-optimal solution where every row and column chooses a vector from the given alphabet set. 
The alphabet-reduction algorithm is completed, so the standard finite-domain-CSP algorithm (sample columns, exhaustively guess the values of the sampled columns, greedily decide the value of each row based on the sampled columns, and greedily decide the value of each column based on all the rows) will result in a PTAS.

%% file: LRA-min-ptas-Final.tex
\section{A polynomial-time multiplicative approximation scheme}
\label{sec:min}

In this section, we prove our main theorem, which gives a multiplicative PTAS for \lra. 

\min*

\begin{proof}
Let $A \in \R^{n_R \times n_C}$ be the input matrix. 
Let $U W^T$ be an optimal solution where $U \in \R^{n_R \times k}$ and $W \in \R^{n_C \times k}$. 
We assume that the rank of $UW^T$ is exactly $k$; otherwise one  can solve for a smaller rank (there are only $k+1$ possibilities).
Let $u_i$ be the $i$th row of $U$ and $w_i$ be the $i$th row of $W$ (as column vectors).
Let $OPT = | \{ (i, j) : A_{i, j} \neq \langle u_i, w_j \rangle \} |$ be the number of {\em errors} that the optimal solution makes.

We will define several constants depending on $\eps_0$ and $k$ (and each other) and see their dependencies at the end of the proof. Let $\delta_0 := 1/20k$ be such a constant.
We first partition $[n_C]$ to $\ell_C$ {\em layers} $J_1, \dots, J_{\ell_C}$ with $\ell_C \leq k$. The desired properties for the layers (called the {\em layer properties}) are as follows.
For $J \subseteq [n_C]$, let $W(J) := \{ w_q : q \in J \}$ be the set of vectors corresponding to $J$ as a multiset (so that $|W(J)| = |J|$). 

\begin{enumerate}[i.]
\item $J_1, \dots, J_{\ell_C}$ partition $[n_C]$ with $|J_j| \leq (2k\delta_0) |J_{j-1}| = |J_{j-1}|/10$. 
\item $\emptyset = T_{J, 0} \subseteq T_{J, 1} \subseteq \dots \subseteq T_{J, {\ell_C}} = \R^k$ are subspaces such that for any $j \in [\ell_C]$, $W(J_j) \subseteq T_{J, j}$. 
\item For any $j \in [\ell_C]$, no subspace $T'$ with $T_{J, {j-1}} \subseteq T' \subsetneq T_{J, j}$ satisfies $|W(J_j) \cap T'| > (1 - \delta_0) |J_j|$.
\end{enumerate}

The following lemma shows the existence of such layers. Note that the layers are used only in the analysis so that the algorithm does not need to construct them. 

\begin{lemma}
There exist $J_1, \dots, J_{\ell_C}$ with $\ell_C \leq k$ satisfying the conditions above. 
\end{lemma}
\begin{proof}
For $j = 1, \dots$, 
we will maintain that $M_j \subseteq [n_C]$ and subspace $T_{J, j} \subseteq \R^k$ that satisfies 
$W(M_j) = W([n_C]) \setminus T_{J, j-1}$ for each $j \geq 1$. 
Let $M_1 = [n_C]$, $T_{J, 0} = \emptyset$, and $j = 1$, and run the following algorithm.

\begin{enumerate}
\item Call a subspace $T$ with $T_{J, j-1} \subsetneq T \subseteq \R^k$ {\em full} if no proper subspace $T_{J, j-1} \subseteq T' \subsetneq T$ has $|T' \cap W(M_j)| \geq (1 - \delta_0)|T \cap W(M_j)|$. (Note that every $T$ with $\dim(T) \leq \dim(T_{j-1}) + 1$ is full, since $|T_{J, j-1} \cap W(M_j)| = 0$.)

\item Choose a full subspace $T_{J, j}$ that contains the maximum number of vectors from $W(M_j)$. Let $J_j \subseteq M_j$ be the set of columns whose vectors belong to $T_{J, j}$, and $M_{j+1} = M_j \setminus J_j$. In particular, $W(J_j) = (T_{J, j}  \setminus T_{J, j - 1}) \cap W([n_C])$. 

\item If $M_{j+1} = \emptyset$, halt. Otherwise, $j \leftarrow j + 1$ and go to Line 1. 
\end{enumerate}
Since the dimension of $T_{J, j}$ is strictly increasing in each iteration, the above algorithm halts before $j = \ell_C \leq k$ iterations. Let us check that $J_j$'s and $T_{J, j}$'s satisfy the three properties for the layers. 
\noindent{\bf Property i.} 
By construction $J_1, \dots, J_{\ell_C}$ partition $[n_C]$. 
Note that also $|J_j| \geq |M_j|(1 - \delta_0)^k \geq |M_j|(1 - k \delta_0) \geq \frac{1 - k \delta_0}{k\delta_0} |J_{j+1}| \geq \frac{1}{2k \delta_0} |J_{j+1}|$; 
one can show the existence of a full subspace $T_{J, j}$ with large $|J_j|$ by starting from $T = \R^k$, and if $T$ is not full, recursively going to the a strict subspace by losing a factor $(1 - \delta_0)$ in size.

\noindent{\bf Property ii.} 
Since we maintain $W(M_i) \cap T_{J, j-1} = \emptyset$ and define $J_j$ be the set of columns of $M_j$ whose vectors belong to $T_{J, j}$, we have $W(J_j) \subseteq T_{J, j}$ for every $j$. 

\noindent{\bf Property iii.} 
Since we chose $T_{J, j}$ to be a full space given $M_j$ and $T_{J, j-1}$, there is no $T'$ with $T_{J, j-1} \subseteq T' \subsetneq T_{J, j}$ that contains more than a $(1 - \delta_0)$ fraction of $J_j$. 
\end{proof}

Apply the lemma for the rows as well to get the partition $I_1, \dots, I_{\ell_R}$ with the subspaces $T_{I, 1}, \dots, T_{I, \ell_R}$ for some $\ell_R \leq k$. 
Call $B_{i,j} := I_i \times J_j$ the $(i, j)$th block. Then we have $\ell_R \times \ell_C$ {\em blocks}.
By guessing, suppose the algorithm knows all block sizes and $OPT$. (There are $n^{O(k)}$ possibilities). 

Let $\eps_1 > 0$ be a constant to be determined. 
Call $B_{i, j}$ {\em clean} if $|I_{i}| |J_{j}| > OPT / \eps_1$, {\em dirty} if $|I_{i}| |J_{j}| < \eps_1 OPT$, and {\em half-clean} otherwise.
Since the algorithm guessed $OPT$ and the block sizes, the algorithm knows which blocks are clean/dirty/half-clean.

Now we describe our algorithm. Note that currently the algorithm has no information about the initial optimal solution $U$ and $W$. 
While describing the algorithm, we will also transform $U$ and $W$ such that 
(1) the transformed solution is still near-optimal, and (2) the algorithm acquires more information about them as it proceeds. 

We divide the algorithm into four phases. 
The goal of the first three phases is to construct a polynomial-size {\em alphabet set} $\Sigma_p \subseteq \R^k$ for each row $p \in [n_R]$ and 
$\Sigma_q \subseteq \R^k$ for each column $q \in [n_C]$ such that there is a near-optimal solution where each row and column gets its vector from its alphabet set;
this corresponds to {\em reducing to a finite-alphabet CSP}. Then the fourth phase obtains a PTAS similarly to finite-alphabet CSPs. 

\paragraph{Phase 1: Obtaining initial samples.} 
Let $\delta_1$ and $t_0$ be other constants to be determined. 
For each $j \in [\ell_C]$, {\em sample} $s_{j, 1}, \dots, s_{j, t_0}$ uniformly at random from $J_j$; throughout the proof, the sampling is only used to show the existence of good samples. The algorithm will enumerate all possible choices of samples. (As the algorithm does not know $J_j$, it cannot perform the actual sampling.)
Similarly, for each $i \in [\ell_R]$, sample $r_{i, 1}, \dots, r_{i, t_0}$ uniformly at random from $I_i$. Call them {\em initial samples}. We will call them {\em perfect} if the following four bad events do no occur.
\begin{enumerate}
\item $A_{r_{i,p}, s_{j,q}} \neq \langle u_{r_{i,p}}, w_{s_{j,q}} \rangle$ for some clean $B_{i, j}$ and $p, q \in [t_0]$; it happens with probability at most $\eps_1 k^2  t_0^2$. Therefore, by requiring that 
\begin{equation}
\eps_1 k^2 t_0^2 \leq \delta_1 / 4,
\label{eq:param2}
\end{equation}
this event happens with probability at most $\delta_1/4$.

\item For some $i \in [\ell_R]$, the sampled row vectors $\{ u_{r_{i', p}} \}_{i' \in [i], p \in [t_0]}$ do not contain a basis of $T_{I, i}$. We will ensure
\begin{equation}
t_0 \geq 2 k^2\log(1/\delta_1)/\delta_0, 
\label{eq:param3}
\end{equation}
which implies that this bad event happens with probability at most $\delta_1/4$; 
assuming that the currently sampled points span some subspace $T'$ with $T_{I, i-1} \subseteq T' \subsetneq T_{I, i}$, the fullness of $T_{I, i}$ ensures that the probability of a new sample from $I_i$ strictly increasing the dimension of the span is at least $\delta_0$, so that the probability that a group of $2k \log(1/\delta_1) / \delta_0 $ samples do not increase the dimension is at most 
\[
(1 - \delta_0)^{(2k \log(1/\delta_1) /\delta_0)} \leq  e^{-(2k \log(1/\delta_1))} = \delta_1^{2k}.
\]
The union bound over $k$ groups 
(note that $t_0$ is large enough to have $k$ separate groups even within a single $I_i$) shows that the overall failure probability is at most $k \delta_1^{2k} \leq \delta_1/4$ whenever $\delta_1$ is smaller than some universal threshold. 

\item Same as 2, but for columns. Similarly, the failure probability is at most $\delta_1/4$. 

\item There exists $j \in [\ell_C]$ and an affine subspace $T'$ of $T_{J, j}$ such that 
\[
\bigg|
\frac{  |\{ q \in J_j : w_q \in T' \} | }{|J_j|} -  \frac{|\{ q \in [t_0] : w_{s_{j, q}} \in T' \} |  }{t_0} 
\bigg|
> \tau
\]
for some constant $\tau$ to be determined. 
If it does not happen, say that the samples $\{ s_{j, q} \}_{j, q}$ are $\tau$-good for the layers $T_{J, 1}, \dots, T_{J, \ell_C}$. 
There are infinitely many affine subspaces, but their VC dimension is at most $k + 1$; 
if $S \subseteq \R^k$ is any set of $k + 2$ points, $S$ cannot be shattered by affine subspaces, because there exists $x \in S$ that can be expressed as an affine combination of $S' \subseteq S \setminus \{ x \}$, which means that no affine subspace can contain $S'$ and exclude $x$. 
Therefore, by the standard sampling guarantee for set systems with bounded VC dimensions~\cite{FM06}, by taking 
\begin{equation}
t_0 = \Omega \big( \frac{k}{\tau^2} \log (1/\tau) + \log(1 / \delta_1) \big),
\label{label:param11}
\end{equation}
one can ensure that the samples are $\tau$-good with probability at least $1 - \delta_1 / 4$.

\end{enumerate}

For the rest of Phase 1 and Phase 2, we condition on the event that the initial samples are perfect. 
Let $d_i = \dim(T_{I, i}) - \dim(T_{I, i - 1})$ with $d_1 := \dim(T_1)$. 
By guessing (at most $2^{O(k t_0 \log t_0)}$ choices) and reordering, the algorithm knows that for every $i \in [\ell_R]$, 
$\{ r_{i', p} \}_{i' \in [i], p \in [d_{i'}]}$ forms a basis of $T_{I, i}$. 
Order these $k$ rows $\{ r_{i', p} \}_{i' \in [\ell_R], p \in [d_{i'}]}$ lexicographically in terms of $(i', p)$ and call them $r_1, \dots, r_k \in [n_R]$. 
Let us denote by $(e_i)_{i \in [k]}$ the standard basis of $\mathbb{R}^k$, and let $C \in \R^{k \times k}$ be the invertible matrix such that $(u_{r_i})^T C = (e_i)^T$ for $i \in [k]$, 
and let $U \leftarrow U \times C$ and $W \leftarrow W \times (C^{-1})^T$. Then the new $U, W$ still remain an optimal solution. 
Note that this transformation also ensures that $T_{I, i} = \{ \spann(e_1, \dots, e_{\dim(T_{I, i})}) \}$. 

This means that the algorithm knows $u_{r_i}$ for every $i \in [k]$ and $T_{I, i}$ for every $i \in [\ell_R]$. 
Moreover, for any $i, j \in [\ell_R] \times [\ell_C]$ with clean $B_{i,j}$, 
using the fact that $A_{r_{i, p}, s_{j, q}} = \langle u_{r_{i, p}}, w_{s_{j, q}} \rangle$ for every $p, q \in [t_0]$, 
the algorithm can determine the projection of $w_{s_{j, q}}$'s to $T_{I, i}$ (denoted by $w_{s_{j, q}}|_{T_{I, i}}$). 
Let $T_{i,j}$ be $T_{J, j}$ projected to $T_{I, i}$. 
Using the fact that the column samples also span their respective subspaces, this also implies that the algorithm knows $T_{i,j}$ for all clean $B_{i,j}$. 

\paragraph{Phase 2: Handling clean hyperblocks.}
Let $\delta_2 > 0$ be a constant to be determined. 
Say $J_j$ and $J_{j+1}$ are {\em super-separated} if $|J_{j+1}|  < \delta_2 |J_{j}|$.
Call $J_i \cup J_{i+1} \cup \dots \cup J_{j}$ a {\em superlayer} 
if $J_i$ is super-separated from $J_{i-1}$ (or $J_i = 1$), 
$J_j$ is super-separated from $J_{j+1}$ (or $J_j = \ell_C$),
and no $J_{j'}$ and $J_{j'+1}$ are super-separated for $j' = i, \dots, j - 1$; 
in words, it is the union of maximally contiguous non-super-separated layers. 
Define similarly for rows. 
Then a {\em superblock} is the product of a row superlayer and a column superlayer. For instance, if rows and columns have $t_R$ and $t_C$ super-separations respectively, the number of superblocks will be $(t_R + 1) \times (t_C + 1)$. 
Call a superblock {\em clean} if all blocks there are clean, {\em dirty} if all blocks are dirty, and {\em half-clean} otherwise. 
We will ensure
\begin{equation}
\delta_2 < \eps_1^2
\label{eq:delta_s_1}
\end{equation}
so that if a superblock $\calB \subseteq [n_R] \times [n_C]$ is half-clean, then any superblock {\em dominated by it} is dirty; a superblock $\calB'$ is dominated by $\calB$ if for any $(i, j) \in \calB$ and $(i', j') \in \calB'$, $i < i'$ and $j < j'$. 

Consider the set of superblocks that are half-clean. Say two superblocks are {\em adjacent} if they share the set of rows or the set of columns. 
Finally, consider a {\em connected component} of half-clean superblocks with this definition of adjacency. 
For each such component, create a {\em half-clean hyperblock} whose row (column) set is the union of all the row (column) sets of its superblocks. 
It is clear from the construction that each row and column belongs to at most one half-clean hyperblock. (I.e., their row sets, possibly with the set of rows with all clean blocks and the set of rows with all dirty blocks, partition $[n_R]$). Use these partitions of rows and columns to create other hyperblocks too, and call them {\em clean} ({\em dirty}) if all blocks are clean (dirty). See Figure~\ref{F:blocks} for an example. 

\begin{figure}
    \centering
	\includegraphics[width=12cm]{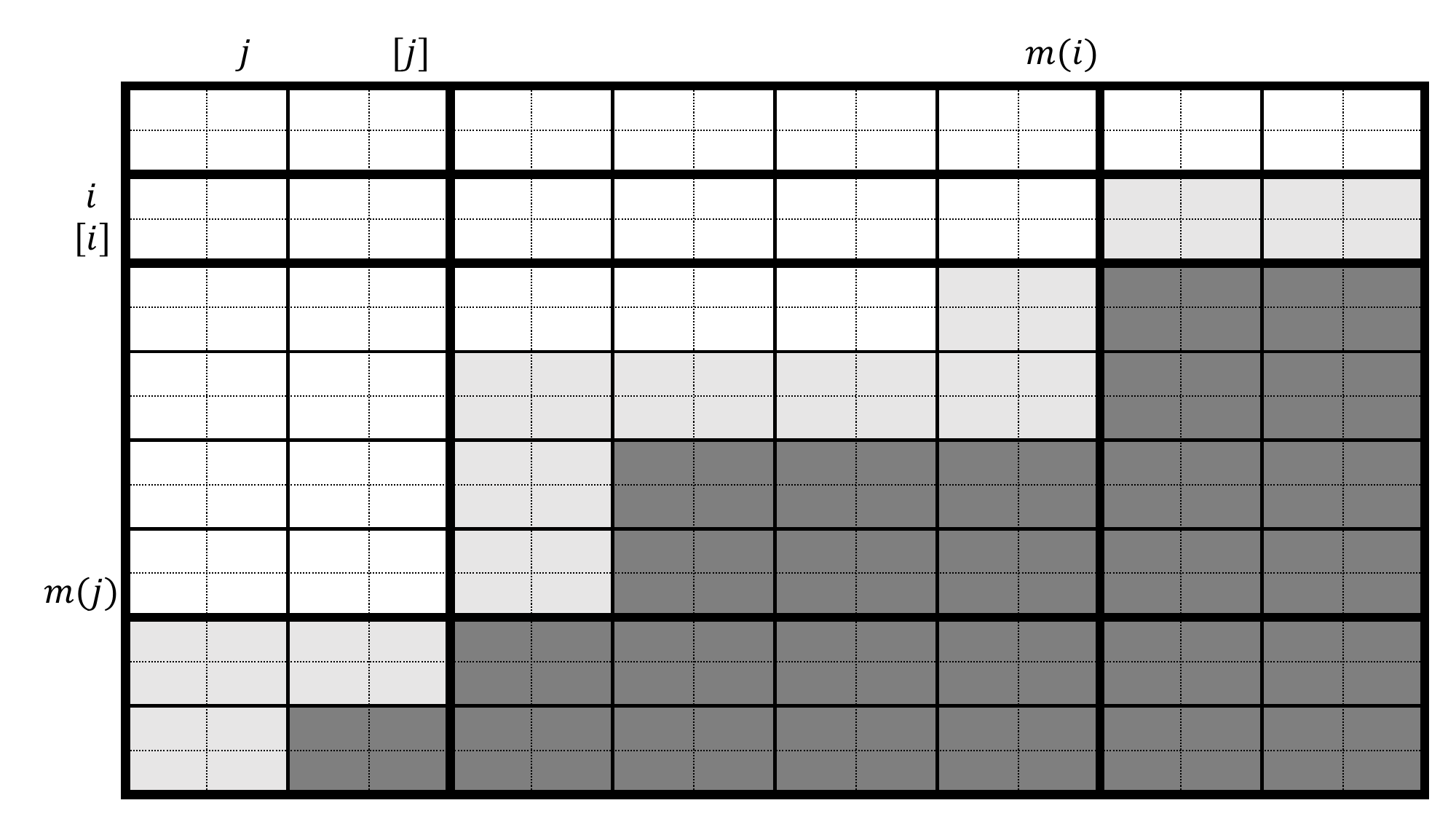}

    \caption{Blocks, superblocks, and hyperblocks are divided by dotted, thin solid, and thick solid lines respectively. The number of blocks, superblocks, and hyperblocks are $16 \times 16$,  $8 \times 8$, and $4 \times 3$ respectively. Clean/dirty/half-cleanness are drawn with respect to superblocks; empty cells are clean, half-shaded cells are half-clean, and full-shaded cells are dirty.}
    \label{F:blocks}
\end{figure}

Of course, a half-clean hyperblock may contain a clean superblock that might be larger than $OPT$. But the following simple claim shows that its size can be still bounded. 

\begin{claim}
Let $\mathfrak{B}$ be a half-clean hyperblock. Then its size is at most $OPT / (\delta_2^{O(k^2)} \eps_1^{O(k)})$.
\label{claim:half-clean}
\end{claim}
\begin{proof}
Let $J'_1, \dots, J'_{d_C}$ and 
$I'_1, \dots, I'_{d_R}$ be the layers comprising $\mathfrak{B}$, ordered as usual from left to right and top to bottom (e.g, $\frakB = 
( \cup_{i \in [d_R]} I'_i ) \times ( \cup_{j \in [d_C]} J'_j)$).
We will show that $|J'_{d_C}| \geq |J'_1| \delta_2^{O(k^2)} (\eps_1)^k$. 
Applying this to the rows and observing that some blocks are half-clean or dirty implies the claim. 

If $J'_j$ and $J'_{j+1}$ are in the same superlayer, their sizes differ by a factor at most $1/\delta_2$. 
If $J'_j$ and $J'_{j+1}$ are in different superlayers, there exist adjacent half-clean superblocks $\calB = I'_i \times J'_j$ and $\calB' = I'_i \times J'_{j+1}$ within $\frakB$ who share the rows. (So $\calB$ is {\em on the left of} $\calB'$.) Let $B$ be the {\em bottom-right} (or smallest) block in $\calB$, and 
$B'$ be the {\em top-left} (or largest) block in $\calB'$. Since both $\calB$ and $\calB'$ are half-clean, $B$ is dirty or half-clean and $B'$ is clean or half-clean, which means that the size of $B'$ is at least $\eps_1^2$ times the size of $B$. Since their superblocks share the set of rows, the number of rows of $B$ and $B'$ differ by a factor at most $1/ \delta_2^k$, which means that the number of columns of $B'$ is at least $\delta_2^k \eps_1^2$ times that of $B$. 
\end{proof}

For $i \in [\ell_R]$, let $m(i) \in [\ell_C]$ be the largest $j$ such that $B_{i, j}$ belongs to a clean hyperblock ($0$ if it belongs to no clean hyperblock).
Also, given $i \in [\ell_R]$, let $[i] \in [\ell_R]$ be the largest index such that $m(i) = m([i])$. 
For example, if $i, i' \in [\ell_R]$ belongs to the same hyperblock, then $[i] = [i']$ and $m(i) = m(i')$. 
Define $m(j)$ and $[j]$ for $j \in [\ell_C]$ symmetrically. 
Note that, by guessing the cleanness and size of each block, the algorithm already knows $m(i)$ for each row and column.

For each row $p \in [n_R]$ and $i \in [\ell_R]$, we let the column samples $\{ s_{j, q} \}_{j \in [m(i)], q \in [t_0]}$ vote for the projection of $u_p$ to $T_{[i], m(i)}$ and call the winner $u_{p, i}$. Formally, 
\begin{equation*}
u_{p, i} := \underset{u \in T_{[i], m(i)}}{\argmin}
\sum_{j = 1}^{m(i)} | \{ q \in [t_0] : A_{p, s_{j, q}} \neq \langle u, w_{s_{j, q}}|_{T_{[i], j}} \rangle \} | \cdot (|J_{j}| / t_0), 
\end{equation*}
where ties are broken arbitrarily. 
It can be computed in time $n^{O(k)}$, by guessing a subset of at most $k$ samples $X$ that satisfy $A_{p, s_{j, q}} = \langle u, w_{s_{j, q}} |_{T_{[i], j}} \rangle$ for all $s_{j, q} \in X$, and trying $u$ that is a solution of the resulting system of linear equations (at most $k$ linearly independent linear equations will uniquely determine $u$).
Since the definition involves only $m(i)$ and $[i]$ instead of $i$, if $i$ and $i'$ belong to the same hyperblock, then $u_{p, i} = u_{p, i'}$, so algorithmically one can compute only one row from each hyperblock. 
(For simplicity, we still treat them separately.)

Having computed $u_{p, i}$'s, the algorithm's ideal situation would be to have $u_{p, i} = u_p|_{T_{[i], m(i)}}$ if $p \in I_i$. 
Of course, this cannot happen always, but we will conduct the following transformation that forces it. 

\begin{itemize}
\item For every $p \in [n_R]$, let $i \in [\ell_R]$ such that $p \in I_i$. 
\begin{itemize}
\item If $u_{p,i}$ is indeed the projection of $u_p$ to $T_{[i], m(i)}$, then do not change anything. 
\item Otherwise, say $p$ is {\em mistaken} and let $u_p \leftarrow u_{p, i}$. By doing this, we (conservatively) make every entry $A_{p, q}$ with $q \in J_j$ and $j > m(i)$ incorrect.
\end{itemize}
\end{itemize}

The heart of this phase is to show that this transformation of $U$ ensures that the solution pair $(U, W)$ is still near-optimal. 
In order to show this, we use the following crucial lemma on the voting. It will be also used in Phase 4. 

\begin{restatable}{lemma}{lemmin}\label{lem:min}
Given parameters $\delta_0$, $\tau < \delta_0 / 100$, $\eps > 0$, $k, \ell \in \N$, there exists $t \leq \poly(k \ell / (\eps \delta_0)) \in \N$ such that the following is true. 
Let $x^*, u_1, \dots, u_n \in \R^k$ where $[n]$ is partitioned into $N_1, \dots, N_{\ell}$ with the associated subspaces $T_1, \dots, T_{\ell} = \R^k$ that satisfy the layer conditions (with parameter $\delta_0$). 
Suppose that the following information is given.
\begin{itemize}
\item $a_p \in \R$ is given for every $p \in [n]$. 
\item $n_i := |N_i|$ for every $i \in [\ell]$. 
\item For each $j \in [\ell]$, $t$ uniformly random samples $s_{j,1}, \dots, s_{j, t} \in N_j$ along with $u_{s_{j,1}}, \dots, u_{s_{j,t}}$.
\end{itemize}
Given the information, $x \in \R^k$ is chosen to be the vector minimizing $\sum_{j=1}^{\ell} (n_j/t) \cdot | \{ p \in [t] : a_{s_{j,p}} \neq \langle x, u_{s_{j,p}} \rangle \} |$.  (Ties are broken arbitrarily.) 
Suppose that the samples are $\tau$-good for $T_1, \dots, T_{\ell}$ with probability at least $1/2$. 
Then, conditioned on the event that they are $\tau$-good, the following holds. 
 
\begin{itemize}
\item $\Pr[x \neq x^*] \leq O \big( \frac{OPT}{\delta_0 n_{\ell}} \big)$, where $OPT := | \{ p \in [n] : a_p \neq \langle x^*, u_p \rangle \} |$. 
\item $\E[| \{ p \in [n] : a_i \neq \langle x, u_p \rangle \}|] \leq (1+ \eps + O(\tau / \delta_0)) OPT$.
\end{itemize}
\end{restatable}

For each $p \in [n_R]$ and $i \in [\ell_R]$, we apply Lemma~\ref{lem:min} with 
$\eps \leftarrow \eps_0 / 16$, 
$\ell \leftarrow m(i)$, 
$N_j \leftarrow J_j$ (so that $[n] \leftarrow J_1 \cup \dots \cup J_{m(i)})$), 
$T_j \leftarrow T_{[i], j}$, 
$\R^k \leftarrow T_{[i], m(i)}$, 
$x^* \leftarrow u_{p}|_{T_{[i], m(i)}}$, and
$u_{s_j, q} \leftarrow w_{s_j, q}|_{T_{[i], j}}$. 
(Since $B_{[i], j}$ is clean for every $j \in [m(i)]$, $w_{s_j, q}|_{T_{[i], j}}$ is known by the algorithm.)
The layer conditions for $T_{[i],j}$ will be satisfied
since $T_{J, j}$'s satisfy the layer conditions, and $T_{[i],j}$'s are their projections to $T_{I, [i]}$.

Note that the lemma assumes we have the correct values of $x^*$ and $u_{s_j, q}$'s, which happens if the initial samples are perfect. 
Therefore, given that the initial samples are perfect, which happens with probability at least $1 - \delta_1$, 
for each $p \in [n_R]$ and $i \in [\ell_C]$ with $p \in I_i$, 
the probability that $p$ is mistaken is at most 
$O(OPT_p/ (\delta_0 |J_{m(i)}|)) / (1 - \delta_1) = O(OPT_p/( \delta_0 |J_{m(i)}|))$, where $OPT_p$ denotes the number of errors in row $p$ before the transformation, which makes the expected number of additional errors in half-clean and dirty hyperblocks due to the transformation bounded by 
\begin{align*}
& \sum_{i \in [\ell_R]} \sum_{p \in I_i} O\bigg(\frac{OPT_p }{ \delta_0 |J_{m(i)}|} \bigg) \bigg( \sum_{j' = m(i)+1}^{\ell_C} |J_{j'}| \bigg)
\leq \sum_{i \in [\ell_R]} \sum_{p \in I_i} O\bigg( \frac{OPT_p }{ \delta_0 |J_{m(i)}|} \delta_2 |J_{m(i)}| \bigg) \\
\leq & \sum_{p \in [n_R]} OPT_p O(\delta_2 / \delta_0) \leq O((\delta_2/\delta_0) OPT),
\end{align*}
where the first inequality follows from the fact that $J_{m(i)}$ and $J_{m(i) + 1}$ are super-separated. (They are in different hyperblocks.)

The second guarantee of the lemma ensures that the total expected errors from clean hyperblocks is at most $(1 + \eps_0/16 + O(\tau / \delta_0)) / (1 - \delta_1)$ times the number of original errors in those blocks. 
Therefore, given that the initial samples are perfect, the expected error of the transformed solution $U, W$ is 
$(1+\eps_0/16 + O(\tau / \delta_0 + \delta_2 / \delta_0)) (1 + 2\delta_1) OPT$. 

Now we do the almost same for columns to compute $w_{q, j}$ for each $q \in [n_C]$ and $j \in [\ell_C]$. 
Sample fresh rows from $I_i$'s, and for each $q \in [n_C]$ and $j \in [\ell_C]$, let them vote for $w_{q, j}|_{T_{m(j), [j]}}$. 
The only difference, which is indeed a simplification, is that we do not worry about the fresh samples being perfect, because when $p \in I_i$ is sampled, then the algorithm can use $u_{p, i} = u_p|_{T_{[i], m(i)}}$ to get the correct projection of $u_p$ to $T_{i, j}$ whenever $B_{i, j}$ is clean. 
This will result in transforming $W$, but the same analysis shows that the 
expected error is at most $(1+\eps_0/8 + O(\tau / \delta_0 + \delta_2 / \delta_0)) (1 + 4\delta_1) OPT$.
Since the algorithm actually tries all possible choices of samples and other relevant information, this analysis implies that for some choice of samples and correct guesses, 
the algorithm computed $u_{p, i}$ and $w_{q, j}$ with the guarantees above.

\paragraph{Phase 3:  Handling half-clean hyperblocks.} 

In Phase 1 and 2, the algorithm constructed a good (partial) vector for each row and column; 
if each $p \in [n_R]$ chooses $u_{p, i} \in T_{[i], m(i)}$ with $p \in I_i$ and 
if each $q \in [n_C]$ chooses $w_{q, j} \in T_{m(j), [j]}$ with $q \in J_j$,
the resulting solution gives a good approximation in clean hyperblocks. 

Consider an entry $(p, q)$ that belongs to a clean hyperblock with $p \in I_i$ and $q \in J_j$,
which means that $[i] \leq m(j)$ and $[j] \leq m(i)$. 
When we {\em extend} $u_{p, i}$ from $T_{[i], m(i)}$ to $T_{I, [i]}$ while ensuring that the projection to $T_{[i], m(i)}$ is preserved, 
and extend $w_{q, j}$ from $T_{m(j), [j]}$ to $T_{J, [j]}$ while ensuring that the projection to $T_{m(j), [j]}$ is preserved, 
we claim that the inner product between them does not depend on the extensions. 
\begin{claim}
If $u_{p, i} \in T_{I, {[i]}}$ and 
$w_{q, j} \in T_{J, [j]}$ with $[i] \leq m(j)$ and $[j] \leq m(i)$, then
$
\langle u_{p, i}, w_{q, j} \rangle 
= \langle u_{p, i}|_{T_{[i], m(i)}}, w_{q, j}|_{T_{m(j), [j]} } \rangle.
$
\label{claim:clean}
\end{claim}
\begin{proof}
For $T \subseteq S \subseteq \R^k$, let $S/T := \{ v \in S : \langle v, u \rangle = 0 \mbox{ for all } v \in T \}$. 
First we write the inner product as 
\[
\langle u_{p, i}, w_{q, j} \rangle = 
\langle u_{p, i}|_{T_{[i], m(i)}}, w_{q, j} \rangle +
\langle u_{p, i}|_{T_{I, [i]} / T_{[i], m(i)}}, w_{q, j} \rangle.
\]
By definition, $u_{p, i}|_{T_{I, [i]} / T_{[i], m(i)}}$ is a vector in $T_{I, [i]}$ orthogonal to every vector in $T_{[i], m(i)}$. 
Since 
$T_{[i], m(i)}$ is the projection of $T_{J, m(i)}$ to $T_{I, [i]}$ and $w_{q, j} \in T_{J, [j]}$ with $[j] \leq m(i)$, the projection of $w_{q, j}$ to $T_{I, [i]}$ belongs to $T_{[i], m(i)}$ as well, which means that $\langle u_{p, i}|_{T_{I, [i]} / (T_{[i], m(i)})}, w_{q, j} \rangle = 0$. So the inner product can be further written as 
\[
\langle u_{p, i}, w_{q, j} \rangle = 
\langle u_{p, i}|_{T_{[i], m(i)}}, w_{q, j} \rangle 
= 
\langle u_{p, i}|_{T_{[i], m(i)}}, w_{q, j}|_{T_{m(j), [j]} } \rangle +
\langle u_{p, i}|_{T_{[i], m(i)}}, w_{q, j}|_{\R^k / T_{m(j), [j]} } \rangle. 
\]
Since 
$T_{m(j), [j]}$ is the projection of $T_{J, [j]}$ to $T_{I, m(j)}$ and $w_{q, j} \in T_{J, [j]}$, it means that 
$w_{q, j}|_{\R^k / T_{m(j), [j]}}$ is indeed equal to 
$w_{q, j}|_{\R^k / T_{I, m(j)}}$ and orthogonal to every vector in $T_{I, m(j)}$. 
Since $u_{p, i}|_{T_{[i], m(i)}} \in T_{I, [i]}$ and $[i] \leq m(j)$, 
we have $\langle u_{p, i}|_{T_{[i], m(i)}}, w_{q, j}|_{\R^k / T_{m(j), [j]} } \rangle = 0$ as well. 
\end{proof}

Now, the algorithm extends this alphabet set for each row and column so that there exists a good solution from the alphabets 
that gives a good approximation in half-clean hyperblocks as well. 
As dirty hyperblocks will be very small compared to $OPT$ and can be totally ignored,
this phase ensures the existence of an overall good approximation solution from the alphabets. 
We use the following result for the additive PTAS for \rlra. 
Recall that in this problem, in addition to $A \in \R^{n_R \times n_C}$, we are additionally given a pair of projection constraints, that is, matrices $P_R$ in $\mathbb{R}^{t_R \times k}$ and $P_C$ in $\mathbb{R}^{t_C \times k}$ as well as real vectors $(a_p)_{p \in [n_R]}$ in $\mathbb{R}^{t_R}$ and $(b_q)_{q \in [n_C]}$ in $\mathbb{R}^{t_c}$, and we require that the matrices $U$ and $W$ also satisfy $P_R(u_p)=a_p$ and $P_C(w_q)=b_q$ for all $p$ and $q$.

\begin{proposition}
For any constants $\eps > 0$ and $k \in \N$, there exist $t = 2^{(1/\eps)^{\poly(k)}}$, $X = n^{(1/\eps)^{\poly(k)}}$ and an algorithm running in time $\poly(X)$ that performs the following task. Given an instance of 
\rlra consisting of $A \in \R^{n_R \times n_C}$, $k \in \N$, and linear constraints
$P_R(u_p)=a_p$ for each row $p \in [n_R]$ and 
$P_C(w_q)=b_q$ for each column $q \in [n_C]$, 
the algorithm outputs $\Sigma_{p, x} \subseteq \R^k$ for each row $p \in [n_R]$, $x \in [X]$ and 
$\Sigma_{q, x} \subseteq \R^k$ for each column $q \in [n_C]$, $x \in X$ 
such that
\begin{enumerate}[(1)]
\item every vector in $\Sigma_{p, x}$ satisfies the linear constraints for row $p$, 
\item every vector in $\Sigma_{q, x}$ satisfies the linear constraints for column $q$, and
\item $|\Sigma_{p, x}|, |\Sigma_{q, x}| \leq t$ for every $x \in [X]$.
\end{enumerate}
For any $S_R \subseteq [n_R]$ and $S_C \subseteq [n_C]$, there exists $x \in [X]$ and $u_p \in \Sigma_{p, x}$ for every row $p \in S_R$ and $w_q \in \Sigma_{q, x}$ for every column $q \in S_C$
such that  
$| \{ (p, q) \in S_R \times S_C : \langle u_p, w_q \rangle \neq A_{p, q} \} | \leq OPT(A_{S_R, S_C}, k) + \eps |S_R| |S_C|$, 
where 
$OPT(A_{S_R, S_C}, k)$ denotes the optimal value for \rlra restricted to the submatrix of $A$ induced by $S_R \times S_C$ with the same linear constraints. 
\label{prop:maxptas-blackbox}
\end{proposition}
\begin{proof}
It directly follows from Proposition~\ref{P:supercore}, which shows the correctness of 
the algorithm in Figure~\ref{A:mainalgo} for \rlra. Here $X$ denotes the number of choices for the supercore, determined by $R_S, C_S, E, M_{R_S}, M_{C_S}$, and $t$ denotes the size of the alphabet, determined by $\Phi \subseteq M_{R_S} \cup M_{C_S}$. 
\end{proof}

Given this proposition, 
for each 
half-clean hyperblock $\frakB = \frakI \times \frakJ$, choose an arbitrary $i \in [\ell_R]$ and $j \in [\ell_C]$ 
such that $I_i \subseteq \frakI$ and $J_j \subseteq \frakJ$, and 
run the above additive PTAS on the matrix $A$
with $k \leftarrow k $, $\eps \leftarrow \eps_2$ 
where each row $p$ has a constraint that its vector belongs to $T_{I, [i]}$ and its projection to $T_{[i], m(i)}$ is $u_{p, i}$ (similarly for columns). 
For each $p \in [n_R]$ and $x \in [X]$, it will create an alphabet set $\Sigma_{p, i, x} \subseteq T_{I, [i]}$
For each $q \in [n_C]$, it will similarly create an alphabet set $\Sigma_{q, j, x} \subseteq T_{J, [j]}$.  

Recall that every row and column belongs to at most one half-dirty hyperblock. 
Therefore, when we consider a solution where each row $p \in [n_R]$ with $p \in I_i$ takes the vector guaranteed by the additive PTAS for the half-clean hyperblock containing $I_i$ (or just $u_{p, i}$ if there is no such hyperblock), and columns take analogous solutions, then the amount of errors from the half-clean hyperblocks is at most $k \eps_2$ times the size of a half-clean hyperblock, which by Claim~\ref{claim:half-clean} is at most $\frac{k \eps_2}{\delta_2^{O(k^2)} \eps_1^{O(k)}} OPT$. 
Note that By Claim~\ref{claim:clean}, this solution extends the solution constructed in Phase 2 and still has the same guarantee in clean hyperblocks. 
By conservatively ignoring all dirty hyperblocks which has an additional cost of $\eps_1 k^2 OPT$, 
The total error for this solution is at most 
\[
OPT \bigg( 1+\eps_0/8 + O\big(\frac{\tau}{\delta_0} + \frac{\delta_2}{\delta_0}\big) + \frac{k \eps_2}{\delta_2^{O(k^2)}\eps_1^{O(k)}} + k^2 \eps_1 \bigg)
\bigg( 1 + 4\delta_1 \bigg)
\]
By ensuring that 
\begin{equation}
\bigg( 1+\eps_0/8 + O\big(\frac{\tau}{\delta_0} + \frac{\delta_2}{\delta_0}\big) + \frac{k \eps_2}{\delta_2^{O(k^2)}\eps_1^{O(k)}} + k^2 \eps_1  \bigg)
\bigg( 1 + 4\delta_1 \bigg) \leq 1 + \eps_0,
\label{eq:param4}
\end{equation}
we get a $(1+\eps_0)$-approximation. 
To satisfy the dependencies between parameters, we set the parameters as follows. 

\begin{itemize}
\item Parameters: $k$, $\eps_0$, $\delta_0$ (layers), $\tau$ (goodness of samples), $t_0$ (sample size), $\delta_1$ (failure probability), 
$\eps_1$ (clean/dirtiness of blocks), $\eps_2$ (additive PTAS guarantee), $\delta_2$ (super-separation).
\item \eqref{eq:param2}: $\eps_1 k^2 t_0^2 \leq \delta_1 / 4$.
\item \eqref{eq:param3}: $t_0 \geq 2 k^2 \log(1/\delta_1) /  \delta_0$.
\item \eqref{label:param11}: $t_0 \geq \Omega \big( \frac{k}{\tau^2} \log (1/\tau) + \log(1 / \delta_1) \big)$. 
\item $t_0 \geq \poly(k / (\delta_0 \eps_0))$ to ensure that  Lemma~\ref{lem:min} works. 
\item \eqref{eq:delta_s_1}: $\delta_2 < \eps_1^2$.
\item \eqref{eq:param4}: 
$\bigg( 1+\eps_0/8 + O\big(\frac{\tau}{\delta_0} + \frac{\delta_2}{\delta_0}\big) + \frac{k \eps_2}{\delta_2^{O(k^2)}\eps_1^{O(k)}} + k^2 \eps_1 \bigg)
\bigg( 1 + 4\delta_1 \bigg) \leq 1 + \eps_0$,
\item Given $\eps_0$ and $k$, fix $\delta_0 = 1/20k$, $\delta_1 = \eps_0 / 100$, and $\tau = O(\eps_0 / \delta_0)$ so that the $O(\tau / \delta_0)$ term in~\eqref{eq:param4} is at most $\eps_0 / 100$. 
\item Choose $t_0 = \poly(k / \eps_0)$ large enough to satisfy \eqref{eq:param3},~\eqref{label:param11} and the condition of Lemma~\ref{lem:min}
\item Choose $\eps_1 = \poly(k / \eps_0)^{-1}$ small enough to satisfy \eqref{eq:param2} and the $k^2 / \eps_1$ term in~\eqref{eq:param4} is at most $\eps_0 / 100$. 
\item Choose $\delta_2 = \poly(k / \eps_0)^{-1}$ small enough to satisfy \eqref{eq:delta_s_1} and 
the $O(\delta_2 / \delta_0)$ term in~\eqref{eq:param4} is at most $\eps_0 / 100$. 
\item Choose $\eps_2 = (2^{\poly(k/\eps_0)})^{-1}$ small enough to satisfy \eqref{eq:param4}. 

\end{itemize}

\paragraph{Phase 4: Finishing off.} 
Let $\frakB_1, \dots, \frakB_{\ell}$ be the half-clean hyperblocks where we run the additive PTAS ($\ell \leq k$).
Let $(i_1, j_1), \dots, (i_{\ell}, j_{\ell}) \in [\ell_R] \times [\ell_C]$ be the {\em representatives} for those hyperblocks; for each $e \in [\ell]$, $(i_e, j_e) \in \frakB_e$ and we ran the additive PTAS after imposing linear constraints according to $i_e$ and $j_e$. 
For $e \in [\ell]$, let $x_e \in [X]$ be the good index guaranteed by Proposition~\ref{prop:maxptas-blackbox} for $\frakB_e$. 
The correct $x_1, \dots, x_{\ell}$ can be guessed from at most $|X|^k$ choices.

Once the algorithm has the correct $x_1, \dots, x_{\ell}$, 
it is guaranteed that there is a constant-size {\em alphabet set} $\Sigma_p := \cup_{e \in [\ell]} \Sigma_{p, i_e, x_e} \subseteq \R^k$ for each row $p \in [n_R]$ and 
$\Sigma_q := \cup_{e \in [\ell]} \Sigma_{q, j_e, x_e} \subseteq \R^k$ for each column $q \in [n_C]$ such that there is a $(1+\eps_0)$-approximate solution where each row and column gets its vector from its alphabet set. 
Then the final phase of the algorithm is again based on sampling; 
sample $s_{j, 1}, \dots, s_{j, t_0}$ from $J_j$ as usual, and guess its correct vector from its alphabet set,
and use them to vote for the rows. 
The crucial difference is that, while we sample from $J_j$ that relies on the layer structure, since we know that the guaranteed solution is globally good, each row gets voted just once (as opposed to $\ell_R$ times before) about its correct position in the entire space $\R^k$. 
Formally, after guessing $\{ w_{s_{j, q}} \}_{j \in [\ell_C], q \in [t_0]}$ from the samples' own alphabet sets, for each $p \in [n_R]$ we apply Lemma~\ref{lem:min} with $\eps \leq \eps_0, \ell \leftarrow \ell_C, N_j \leftarrow J_j, T_j \leftarrow T_j, \R^k \leftarrow \R^k, u_{s_j, q} \leftarrow w_{s_j, q}$. 
Then the lemma guarantees the expected value of the final solution obtained is at most a factor $(1+\eps_0)$ worse than the intended solution. 
(We do not need the first guarantee about the chosen vector being equal to $x^*$. Also note that the vectors for rows might be outside their alphabets.)
Therefore, for some choice of samples and their positions from their alphabets, the total error of the computed solution is at most $(1 + \eps_0)^2 OPT$, finishing the proof of the theorem.

\paragraph{Total running time.}
It can be checked that the the total number of guesses that Phase 1 and Phase 2 make is $n^{\poly(k/\eps_0)}$.
Given our choice of $\eps_2 = (2^{\poly(k/\eps_0)})^{-1}$, the total running time spent on the additive PTASes 
is $n^{(1/\eps_2)^{\poly(k)}} = n^{2^{\poly(k/\eps_0)}}$, which is also an upper bound on $|X|$. 
The running time in Phase 4 is dominated by the time to guess the correct $x_1, \dots, x_{\ell}$, which is at most $|X|^k \leq n^{2^{\poly(k/\eps_0)}}$. 
Therefore, the total running time is at most $n^{2^{\poly(k/\eps_0)}}$.
After Phase 3, each vector in $\Sigma_p$ is described by the Thom encoding with bit complexity at most $\poly_{k,\varepsilon}(\tau)$, where $\tau$ is the bit complexity of the input matrix $A$, and the degree of the polynomial is independent of $\varepsilon$ and $k$ . 
As mentioned in the definition of $u_{p, i}$ before the statement of Lemma~\ref{lem:min}, the finally chosen vector for each row and column will be the solution to a system of at most $\poly(k)$ linear equations whose coefficients are from $\Sigma_p$ and $\Sigma_q$. Therefore, the final vectors can be computed using the real algebraic solvers described in Section~\ref{sec:algeq} and be described by Thom encodings with bit complexity $\poly_{k,\varepsilon} (\tau)$, where the degree of the polynomial is independent of $\varepsilon$ and $k$.
\end{proof}

\subsection{Proof of Lemma~\ref{lem:min}}
In this section, we prove Lemma~\ref{lem:min}, which is restated below.

\lemmin*

\begin{proof}
Let $S_i = \{ s_{i, 1}, \dots, s_{i, t} \}$ be the set of samples from $N_i$. 
For any $y \in \R^k$ and $i \in [\ell]$, let us define the following three quantities.

\begin{itemize}
\item $\alpha_i(y) := | \{ q \in N_i : \langle u_{q}, y \rangle \neq \langle u_{q}, v \rangle \} |$: the number of indices where $x^*$ and $y$ {\em behave} differently. 

\item $\beta_i(y) := | \{ q \in N_i : \langle u_{q}, y \rangle \neq a_{ q} \} |$: the number of {\em errors} when we choose $y$. 

\item $\gamma_i(y) := 
| \{ q \in S_i : \langle u_{q}, y \rangle \neq a_{ q } \} |
\cdot (n_i/t)$: the estimate of $\beta_i(y)$ using the samples $s_{i,1}, \dots, s_{i,t}$.
\end{itemize}
If $y$ is clear from context, let us just write $\alpha_i, \beta_i, \gamma_i$, and let $\alpha, \beta, \gamma$ be the sum of $\alpha_i$'s, $\beta_i$'s, $\gamma_i$'s respectively. And let $\alpha^*_i, \beta^*_i, \gamma^*_i$ be the quantities when $y = x^*$.
Of course, we only see $\gamma_i(y)$.

For $y \in \R^k$ and  $i \in [\ell]$, if $y - x^*$ is orthogonal to $T_i$ (including the case $y = x^*$), say $y$ is {\em fortunate} in $i$. We have $\langle y , u_{q} \rangle = \langle x^*, u_{q} \rangle$ for every $q \in N_i$, so $y$ and $x^*$ behave the same with respect to $N_i$. So, $\alpha_i(y) = \alpha^*_i, \beta_i(y) = \beta^*_i, \gamma_i(y) = \gamma^*_i$.

When $y - x^*$ is not orthogonal to $T_i$, let $T_y$ be the proper subspace of $T_i$ orthogonal to $y - x^*$, and $N_y \subseteq N_i$ be the set of indices whose vectors are in $T_y$. Note that $\alpha_i^* = 0$ by definition, and $\alpha_i(y) = n_i - |N_y|$. 
If $i$ is the smallest index such that 
$y - x^*$ is not orthogonal to $T_i$, the fullness of $T_i$ implies that 
$|N_y| \leq (1 - \delta_0) n_i$, which implies that $\alpha_i(y) \geq \delta_0 n_i$ and $\beta_i(y) \geq \max(0, \alpha_i(y) - \beta_i^*)$. 
(If $|N_y| > (1 - \delta_0) n_i$ and $y - x^*$ is orthogonal to $T_{i-1}$, then $T_y$ satisfies $T_{i-1} \subseteq T_y \subsetneq T_i$ and contains most of $N_i$, which contradicts the fullness of $T_i$.)

For a set $E \subseteq N_i$, let 
$\gamma_i(E) := |E \cap S_i| \cdot (n_i / t)$. Intuitively, 
$\gamma_i(E)$ is the estimate of $|E|$ via samples. 
If the samples are $\tau$-good, as $N_y$ is a subset of $N_i$ contained in a strict subspace of $T_{i}$, 
$| |N_y| - \gamma_i(N_y)| \leq \tau |N_i|$. 
Let $W_i = \{ q \in N_i : \langle v, u_q \rangle \neq a_q \}$ be the set of indices where $x^*$ is {\em wrong} (e.g., $|W_i| = q^*_i$).
Now let $x \neq x^*$ be the chosen vector, which means $\gamma(x) \leq \gamma^*$.

\paragraph{First claim.} 
We first prove the first claim, which shows that the conditional probability of $x \neq x^*$ given the samples are $\tau$-good is 
at most $O ( \frac{OPT}{\delta_0 n_{\ell}})$. 
If $\beta^* = OPT \geq 0.01 \delta_0 n_{\ell}$, the trivial probability bound of $1$ proves the claim, so assume that $\beta^* < 0.01 \delta_0 n_{\ell}$.
The proof considers the event that there exists $i \in [\ell]$ with $\gamma_i(W_i) > 0.1 \delta_0 n_i$;
we will show that this event will happen with a small probability, and if it does not happen, our output $x$ must be equal to $x^*$. 

\begin{claim}
Suppose that $\gamma_i(W_i) \leq 0.1 \delta_0 n_i$ for all $i \in [\ell]$. Then $x = x^*$. 
\end{claim}
\begin{proof}
Assume towards contradiction that $x \neq x^*$ and let $f$ be the first index that $T_f$ is not orthogonal to $x - x^*$. As before, let $N_x = \{ q \in N_f : \langle u_q, x \rangle = \langle u_q, x^* \rangle \}$. 
Then the fullness of $T_f$ ensures that $|N_x| \leq (1 - \delta_0) n_f$. 
Since the samples are $\tau$-good, 
$\gamma_f(N_f \setminus N_x) \geq (\delta_0 - \tau) n_f  \geq 0.9 \delta_0 n_f$.
Since we assume $\gamma_i(W_i) \leq 0. 1\delta_0 n_i$ for all $i$, then 
\[
\sum_{i = f}^{\ell} \gamma_i(W_i) \leq 
\sum_{i = f}^{\ell} 0.1\delta_0 n_i \leq 
0.2\delta_0 n_f
\]
is the estimated number errors $x^*$ makes in $N_f, \dots, N_\ell$ and 
while the estimated number of errors that $x$ makes in $N_f$ only is at least 
\[
\gamma_f((N_f \setminus N_x) \setminus W_f)
\geq 
\gamma_f(N_f \setminus N_x) - \gamma_f(W_f)
\geq 0.9 \delta_0 n_f - 0.1 \delta_0 n_f = 0.8 \delta_0 n_f,
\]
which leads to contradiction, because $x$ and $x^*$ get the same estimated errors from $N_1, \dots, N_{f-1}$. 
\end{proof}

Finally, we bound the probability of having some $i \in [\ell]$ with 
$\gamma_i(W_i) > 0.1 \delta_0 n_i$. 
Let $\zeta = OPT / n_{\ell}$, which implies that 
$\zeta \geq OPT / n_i$ for all $i \in [\ell]$. 
Before the conditioning, because samples were uniformly and independently sampled from each $N_i$, Chernoff and union bounds imply that this probability is upper bounded by 
\[
\ell \bigg( 
\frac{e}{0.1 \delta_0 / \zeta}
\bigg)^{0.1 \delta_0 t}.
\]
(When we sample from $N_i$, each sample from $i$ is from $W_{i}$ with probability at most $\zeta$ and we bound the probability that the fraction of the samples from $W_i$ is at least $0. 1\delta_0 $.)
For some sufficiently large $t = \Omega(\ell /\delta_0)$, this probability 
is at most $O(\zeta / \delta_0) = O(OPT / (n_{\ell} \delta_0))$. Since we assumed that the samples are $\tau$-good with probability at least $1/2$, the conditional probability is at most $2$ times bigger than this.

\paragraph{Second claim.} 
For each $i \in [\ell]$ and $y \in \R^k$, 
let $W_{i, y} = \{ q \in W_{i, y} : \langle y, u_q \rangle = a_q \}$ be the set of indices where $x^*$ is wrong but $y$ is right. The VC dimension of the family $\{ W_{i, y} \}_y$ for fixed $i$ is $O(k)$. If we define 
\[
\tau'_i := 
\max \bigg(
\max_y{ \frac{| |W_{i, y}| - \gamma_i(W_{i, y}) |}{n_i} }, \quad
\frac{| |W_i| - \gamma_i(W_i) | }{n_i}
\bigg), 
\]
Then, for any $y \in \R^k$, since $\{ q \in N_i : \langle y, u_q \rangle = a_q  \} = (N_y \setminus W_i) \cup W_{i, y}$, 
we have 
$
|\beta_i(y) - \gamma_i(y)|
\leq (\tau + 2\tau'_i)n_i. 
$
Furthermore, the standard sampling guarantee for sets systems with bounded VC dimension~\cite{FM06} guarantees that,
for a fixed $t = \poly(k\ell / (\eps \delta_0))$, 
$\tau'_i$ is a sub-gaussian random variable with the sub-gaussian norm $\| \tau'_i \|_{\psi_2} \leq \poly(\eps \delta_0 / (k \ell))$~\cite{vershynin2018high}.
Therefore, if we let $\tau' := \max_{i \in [\ell]} \tau_i'$, one can ensure $\E[\tau'] \leq \eps \delta_0 / 10000$. 

Let $x$ be the chosen vector and $f$ be the first index that $T_f$ is not orthogonal to $x - x^*$.
Then $F := \{ 1, \dots, f-1\}$ is the set of indices where $x$ is fortunate. 
Let $X := \beta(x) - \beta(x^*)$ be the amount of additional error we incur by choosing $x$. 
Note that $x$, $f$, $\tau'$, and $X$ are correlated random variables. 
We present the following two ways to bound $X$ given $f$.

\noindent {\it First method.}
The first method works well when $n_f$ is small compared to $\beta^*$ (we will apply it when $\beta^* \geq 0.01 \delta_0 n_f$).
Note that $\gamma^*_i = \gamma_i$ when $i \in F$.
For $i = f, \dots, \ell$, using the approximation between $\beta$'s and $\gamma$'s when $i \notin F$, we have 
\[
\sum_{i < f} \gamma^*_i + 
\sum_{i \geq f} (\beta_i - \tau n_i - 2\tau' n_i)  
\leq \gamma(x) \leq \gamma^* \leq 
\sum_{i < f} \gamma^*_i + 
\sum_{i \geq f} (\beta_i^* + \tau n_i + 2\tau' n_i),
\]
which implies that the additional error caused by picking $x$ is at most 
\[
(2\tau + 4\tau') \sum_{i \geq f} n_i \leq 
(4\tau + 8\tau') n_f,
\]
using the fact that $n_i$'s are geometrically decreasing with a factor $k\delta_0 \leq 1/2$. 

\noindent {\it Second method.}
We will apply the second method when $\beta^* < 0.01 \delta_0 n_f$.
In this case, we will conservatively say $x$ makes errors on all $N_{f}, \dots, N_{\ell}$ so that $X \leq 2n_f$;
what we will do here is to bound the probability that $f$ takes such a value. 
Note that $\beta_f + \beta^*_f \geq \delta_0 n_f$ implies 
$\beta_f > 0.99 \delta_0 n_f$. We consider the following two events.

\begin{itemize}
\item The probability of $\gamma_f(W_f) > 0.1 \delta_0 n_f$ is bounded by
\[
2\bigg( 
\frac{e}{(\delta_0 / (10 (\beta_f^* / n_f)))}
\bigg)^{(\delta_0 / 10) t}.
\]
(Before conditioning, 
each sample is from $W_f$ with probability $\beta^*_f / n_f$, and we bound the probability that the fraction of the samples from $W_f$ is at least $(\delta_0 / 10)$.)
By letting $t \geq \Omega(\ell /(\eps \delta_0))$, we can ensure that the probability is at most $(\beta^*/n_f) \cdot (\eps/80\ell)$. 

\item 
If $\gamma_f(W_f) \leq 0.1 \delta_0 n_f$, then the estimated errors of $x$ from $N_f$ is already at least $0.5 \delta_0 n_f$.
Because the total number of entries from $N_{f+2}, \dots, N_{\ell}$ is at most $O((k\delta_0)^2 n_f)$, so given that $\gamma_f(W_f) \leq 0.1\delta_0n_f$, we should get 
$\gamma_{f + 1}(W_{f + 1}) \geq (\delta_0/3)n_{f}$ in order for $x^*$ to have more estimated errors than $x$. 
If $n_{f+1} < (\delta_0/3)n_f$, this can never happen. 
Otherwise, this probability is bounded by 
\[
2\bigg( 
\frac{e}{(\delta_0 / 3) / (\beta_{f+1}^*/n_{f+1})}
\bigg)^{(n_f/n_{f+1}) (\delta_0 / 3)  t}
\leq 
2\bigg( 
\frac{e}{ (n_{f+1}/n_f) (\delta_0 / 3) / (\beta^* / n_f)}
\bigg)^{(n_f/n_{f+1}) (\delta_0 / 3)  t}
\]
(Before conditioning, in $N_{f+1}$, each sample is from $W_{f + 1}$ with probability $\beta^*_{f + 1} / n_{f+1} \leq (\beta^*/n_f) \cdot (n_f / n_{f+1})$, and we bound the probability that the fraction of the samples from $W_{f + 1}$ is at least $(n_f/n_{f+1}) (\delta_0 / 3)$.)
Again, by letting $t \geq \Omega(\ell /(\eps \delta_0))$, one can again ensure that the probability is at most $(\beta^*/n_f) \cdot (\eps/80\ell)$. 

\item Therefore, when $i \in [\ell]$ is such that $\beta^* < 0.01 \delta_0 n_i$, we have $\Pr[f = i] \leq (\beta^*/n_i) \cdot (\eps/ 40 \ell)$. 
\end{itemize}

Finally, we use the above two methods to bound $\E[X]$. 
Let $i'$ be the first index that $0.01 \delta_0 n_{i'} \leq \beta^*$. 
The additional amount of error of $x$ compared to $x^*$ (denoted by $X$) can be bounded by 
\begin{align*}
& \sum_{i=i'}^{\ell} \Pr[f = i] \E[X | f = i]
+ \sum_{i=1}^{i' - 1} \Pr[f = i] \E[X | f = i] \\ 
\leq & \sum_{i=i'}^{\ell} \Pr[f = i] \E[(4\tau + 8\tau') n_i | f = i]
+ \sum_{i=1}^{i' - 1} (\beta^* / n_i)(\eps/40\ell) \cdot (2 n_i) \\
\leq & \sum_{i=i'}^{\ell} n_{i'} \Pr[f = i] \E[(4\tau + 8\tau') | f = i]
+ (\eps / 20) \beta^* \\
=&  n_{i'} \cdot  \E[(4\tau + 8\tau') | f \leq i'] \Pr[f \geq i']
+ (\eps / 20) \beta^*\\
\leq&  n_{i'}  \cdot \E[(4\tau + 8\tau')]
+ (\eps / 20) \beta^*\\
\leq& (200\beta^* / \delta_0) (4\tau + \E[8\tau']) + (\eps / 20) \beta^*
\leq (1000\tau / \delta_0 + 0.1\eps) \beta^*.
\end{align*}
\end{proof}

%% file: LRA-hardness.tex
\section{Hardness of $\ell_0$ Low-Rank Approximation}

In this section, we give the $\Omega(\log n)$-factor hardness of \lra.

\hardness*

Our reduction shall use the following well-known $\Omega(\log n)$-hardness of {\sc Set Cover}.

\begin{theorem}[\cite{FeigeSetCover},\cite{DS14}]		\label{thm:set-cover}
	Let $(U,\calS)$ be an instance of Set Cover with $|U| = n$. Then for $\ell = \ell(n)$, it is NP-hard to distinguish between the following cases:
	\begin{itemize}
		\item {\bf Yes Case}. There exists a choice of $S_{i_1},\ldots,S_{i_\ell} \in \calS$ such that every element of $U$ is covered by exactly one of the sets.
		\item {\bf No Case}. Any set cover of $U$ by $\calS$ is of size at least $\Omega(\ell\cdot\log n)$.
	\end{itemize} 
\end{theorem}

\begin{proof}[Proof of Theorem \ref{thm:hardness}]
	Let $\calI := (U,\calS)$ with $|U| = n$ be an instance of Set Cover from Theorem \ref{thm:hardness}, and let $\ell = \ell(n)$ be as in the statement of the theorem. 
	Note that one can assume $m = \Omega(\ell \cdot \log n)$, because otherwise $\calI$ can never be a NO instance. 
	We then construct our \lra instance from $\calI$ as follows.
	
	\paragraph{Construction.} Given $\calS := \{S_1,\ldots,S_m\}$, let $A \in \{0,1\}^{n \times m}$ be the matrix whose $i^{th}$ column is the indicator vector $\mathbb{I}_{S_i}$ of the $i^{th}$ set. Let $V$ denote the kernel of $A$ and let $b_1,\ldots,b_k$ be any basis of $V$ (where $k$ is the dimension of $V$). Without loss of generality, we may assume that there exists a vector $w$ such that $Aw = (-1,\ldots,-1)$ (otherwise the reduction can simply identify this as a NO case instance). Then we construct our target matrix $M \in \mathbb{R}^{m \times ((m + 1)k + 1)}$ as follows.
	\begin{itemize}
		\item The first $(m + 1)k$ columns of $M$ consist of $(m + 1)$-copies of $b_1,\ldots,b_k$.
		\item The last column is the vector $w$.
	\end{itemize}
	The final instance consists of the matrix $M$, with rank parameter $k$, and $\ell_0$ error parameter $\ell$. We now analyze the reduction.
	
	\paragraph{Completeness.} Suppose $\calI$ is a Yes instance. Then (up to re-ordering), we may assume that the sets $S_1,\ldots,S_\ell$ cover the all the elements, and every element is covered exactly once. Let $v = -\sum_{i \in [\ell]} e_i$. Then $Av = Aw = (-1,\ldots,-1)$ and hence $(w-v) \in V$. Then let $M' \in \mathbb{R}^{m \times ((m + 1)k + 1)}$ be the matrix whose first $(m + 1)k$ columns are identical to that of $M$, and the last column is $w - v$. Then clearly $M'$ is of rank $k$ (since $b_1, \dots, b_k, w - v \in V$). Furthermore, the two matrices only differ in the last column and hence $\|M - M'\|_0 = \|v\|_0 = \ell$.
	
	\paragraph{Soundness.} Suppose $\calI$ is a NO instance. We may assume that there exists a rank-$k$ matrix $M' \in \mathbb{R}^{m \times {(m + 1)k + 1}}$ such that $\|M - M'\|_0 \leq m$ (otherwise we are done) -- let $M'$ denote such a matrix with the smallest $\ell_0$-error. We have the following useful claim:
	\begin{claim}
		The column space of $M'$ must be identical to $V$.
	\end{claim}
	\begin{proof}
		Suppose not, then there exists $i \in [k]$ such that $b_i \notin {\rm col}(M')$ which implies that $\|M' - M\|_0 \geq m + 1$ which contradicts our choice of $M'$.
	\end{proof}
	Now let us denote $\ell' = \|M - M'\|_0$, and let $v$ be the last column of $M'$. Then note that $\|w - v\|_0 \leq \|M - M'\|_0 = \ell'$. Furthermore, using the above claim we know that $v \in V$ and hence $A(w - v) = -(1,\ldots,1)$. But then the non-zero indices of $(w - v)$ must form a valid set cover of $\calS$, which using the No case guarantee of $\calI$ implies that $\ell' \geq \Omega(\ell\cdot\log n)$.
\end{proof}

%% file: main.bbl
\newcommand{\etalchar}[1]{$^{#1}$}
\begin{thebibliography}{ADLVKK03}

\bibitem[ADLVKK03]{alon2003random}
Noga Alon, W.~Fernandez De~La~Vega, Ravi Kannan, and Marek Karpinski.
\newblock Random sampling and approximation of {MAX}-{CSP}s.
\newblock {\em Journal of Computer and System Sciences}, 67(2):212--243, 2003.

\bibitem[AKK95]{arora1995polynomial}
Sanjeev Arora, David Karger, and Marek Karpinski.
\newblock Polynomial time approximation schemes for dense instances of
  {NP}-hard problems.
\newblock In {\em Proceedings of the 27th Annual ACM Symposium on Theory of
  Computing (STOC)}, pages 284--293, 1995.

\bibitem[BBB{\etalchar{+}}19]{ban2019ptas}
Frank Ban, Vijay Bhattiprolu, Karl Bringmann, Pavel Kolev, Euiwoong Lee, and
  David~P. Woodruff.
\newblock A {PTAS} for $\ell_p$-low rank approximation.
\newblock In {\em Proceedings of the 30th Annual ACM-SIAM Symposium on Discrete
  Algorithms (SODA)}, pages 747--766. SIAM, 2019.

\bibitem[BCW20]{bakshi2020robust}
Ainesh Bakshi, Nadiia Chepurko, and David~P. Woodruff.
\newblock Robust and sample optimal algorithms for {PSD} low rank
  approximation.
\newblock In {\em IEEE 61st Annual Symposium on Foundations of Computer Science
  (FOCS)}, pages 506--516. IEEE, 2020.

\bibitem[BCW22]{bakshi2022low}
Ainesh Bakshi, Kenneth~L Clarkson, and David~P. Woodruff.
\newblock Low-rank approximation with $1/\epsilon^3$ matrix-vector products.
\newblock {\em arXiv preprint arXiv:2202.05120}, 2022.

\bibitem[BHHS11]{barak2011subsampling}
Boaz Barak, Moritz Hardt, Thomas Holenstein, and David Steurer.
\newblock Subsampling mathematical relaxations and average-case complexity.
\newblock In {\em Proceedings of the 22nd Annual ACM-SIAM Symposium on Discrete
  Algorithms (SODA)}, pages 512--531. SIAM, 2011.

\bibitem[BKW17]{BKW17}
Karl Bringmann, Pavel Kolev, and David~P. Woodruff.
\newblock Approximation algorithms for $\ell_0$-low rank approximation.
\newblock {\em Advances in Neural Information Processing Systems}, 30, 2017.

\bibitem[BPR10]{roy}
Sugata Basu, Richard Pollack, and Marie-Fran{ç}oise Roy.
\newblock {\em Algorithms in Real Algebraic Geometry}.
\newblock Algorithms and Computation in Mathematics. Springer Berlin,
  Heidelberg, 2010.

\bibitem[BRS11]{barak2011rounding}
Boaz Barak, Prasad Raghavendra, and David Steurer.
\newblock Rounding semidefinite programming hierarchies via global correlation.
\newblock In {\em IEEE 52nd Annual Symposium on Foundations of Computer Science
  (FOCS)}, pages 472--481. IEEE, 2011.

\bibitem[BWZ19]{ban2019regularized}
Frank Ban, David Woodruff, and Richard Zhang.
\newblock Regularized weighted low rank approximation.
\newblock {\em Advances in Neural Information Processing Systems}, 32, 2019.

\bibitem[CFLM22]{cohen2022fitting}
Vincent {Cohen-Addad}, Chenglin Fan, Euiwoong Lee, and {Arnaud de} {Mesmay}.
\newblock Fitting metrics and ultrametrics with minimum disagreements.
\newblock In {\em IEEE 63rd Annual Symposium on Foundations of Computer Science
  (FOCS)}, pages 301--311. IEEE, 2022.

\bibitem[CGK{\etalchar{+}}17]{chierichetti2017algorithms}
Flavio Chierichetti, Sreenivas Gollapudi, Ravi Kumar, Silvio Lattanzi, Rina
  Panigrahy, and David~P. Woodruff.
\newblock Algorithms for $\ell_p$ low-rank approximation.
\newblock In {\em International Conference on Machine Learning}, pages
  806--814. PMLR, 2017.

\bibitem[CLMW11]{candes2011robust}
Emmanuel~J. Cand{\`e}s, Xiaodong Li, Yi~Ma, and John Wright.
\newblock Robust principal component analysis?
\newblock {\em Journal of the ACM (JACM)}, 58(3):1--37, 2011.

\bibitem[COCF10]{coja2010efficient}
Amin Coja-Oghlan, Colin Cooper, and Alan Frieze.
\newblock An efficient sparse regularity concept.
\newblock {\em SIAM Journal on Discrete Mathematics}, 23(4):2000--2034, 2010.

\bibitem[Cox11]{cox}
David~A Cox.
\newblock {\em Galois Theory}, volume~61.
\newblock John Wiley \& Sons, 2011.

\bibitem[CP10]{candes2010matrix}
Emmanuel~J. Candes and Yaniv Plan.
\newblock Matrix completion with noise.
\newblock {\em Proceedings of the IEEE}, 98(6):925--936, 2010.

\bibitem[CW17a]{clarkson2017low}
Kenneth~L. Clarkson and David~P. Woodruff.
\newblock Low-rank approximation and regression in input sparsity time.
\newblock {\em Journal of the ACM (JACM)}, 63(6):1--45, 2017.

\bibitem[CW17b]{clarkson2017lowsoda}
Kenneth~L. Clarkson and David~P. Woodruff.
\newblock Low-rank {PSD} approximation in input-sparsity time.
\newblock In {\em Proceedings of the 28th Annual ACM-SIAM Symposium on Discrete
  Algorithms (SODA)}, pages 2061--2072. SIAM, 2017.

\bibitem[DHJ{\etalchar{+}}18]{dan2018low}
Chen Dan, Kristoffer~Arnsfelt Hansen, He~Jiang, Liwei Wang, and Yuchen Zhou.
\newblock Low rank approximation of binary matrices: {C}olumn subset selection
  and generalizations.
\newblock In {\em 43rd International Symposium on Mathematical Foundations of
  Computer Science}, 2018.

\bibitem[dlVKKV05]{de2005tensor}
W.~Fernandez de~la Vega, Marek Karpinski, Ravi Kannan, and Santosh Vempala.
\newblock Tensor decomposition and approximation schemes for constraint
  satisfaction problems.
\newblock In {\em Proceedings of the 37th annual ACM Symposium on Theory of
  Computing (STOC)}, pages 747--754, 2005.

\bibitem[dlVKM07]{de2007linear}
Wenceslas~Fernandez de~la Vega and Claire Kenyon-Mathieu.
\newblock Linear programming relaxations of maxcut.
\newblock In {\em Proceedings of the 18th Annual ACM-SIAM Symposium on Discrete
  Algorithms (SODA)}, pages 53--61. SIAM, 2007.

\bibitem[DS14]{DS14}
Irit Dinur and David Steurer.
\newblock Analytical approach to parallel repetition.
\newblock In {\em Proceedings of the 46th annual ACM symposium on Theory of
  computing}, pages 624--633, 2014.

\bibitem[ES83]{erdos}
Paul Erd{\H{o}}s and Mikl{\'o}s Simonovits.
\newblock Supersaturated graphs and hypergraphs.
\newblock {\em Combinatorica}, 3(2):181--192, 1983.

\bibitem[Fei98]{FeigeSetCover}
Uriel Feige.
\newblock A threshold of ln n for approximating set cover.
\newblock {\em Journal of the ACM (JACM)}, 45(4):634--652, 1998.

\bibitem[FGL{\etalchar{+}}19]{fomin2019approximation}
Fedor~V. Fomin, Petr~A. Golovach, Daniel Lokshtanov, Fahad Panolan, and Saket
  Saurabh.
\newblock Approximation schemes for low-rank binary matrix approximation
  problems.
\newblock {\em ACM Transactions on Algorithms (TALG)}, 16(1):1--39, 2019.

\bibitem[FGP20]{fomin2020parameterized}
Fedor~V. Fomin, Petr~A. Golovach, and Fahad Panolan.
\newblock Parameterized low-rank binary matrix approximation.
\newblock {\em Data Mining and Knowledge Discovery}, 34:478--532, 2020.

\bibitem[FK96]{frieze1996regularity}
Alan Frieze and Ravi Kannan.
\newblock The regularity lemma and approximation schemes for dense problems.
\newblock In {\em IEEE 37th Annual Symposium on Foundations of Computer Science
  (FOCS)}, pages 12--20. IEEE, 1996.

\bibitem[FLM{\etalchar{+}}17]{fomin2017rigidity}
Fedor~V. Fomin, Daniel Lokshtanov, S.~M. Meesum, Saket Saurabh, and Meirav
  Zehavi.
\newblock {Matrix Rigidity from the Viewpoint of Parameterized Complexity}.
\newblock In Heribert Vollmer and Brigitte Vall{\'e}e, editors, {\em 34th
  Symposium on Theoretical Aspects of Computer Science (STACS 2017)}, volume~66
  of {\em Leibniz International Proceedings in Informatics (LIPIcs)}, pages
  32:1--32:14, Dagstuhl, Germany, 2017. Schloss Dagstuhl--Leibniz-Zentrum fuer
  Informatik.

\bibitem[FM06]{FM06}
Uriel Feige and Mohammad Mahdian.
\newblock Finding small balanced separators.
\newblock In {\em Proceedings of the 38th Annual ACM Symposium on Theory of
  Computing (STOC)}, pages 375--384, 2006.

\bibitem[FRVB18]{fan2018metric}
Chenglin Fan, Benjamin Raichel, and Gregory Van~Buskirk.
\newblock Metric violation distance: {H}ardness and approximation.
\newblock In {\em Proceedings of the 29th Annual ACM-SIAM Symposium on Discrete
  Algorithms (SODA)}, pages 196--209. SIAM, 2018.

\bibitem[Gri76]{grigoriev1976using}
D.~Yu Grigoriev.
\newblock Using the notions of separability and independence for proving the
  lower bounds on the circuit complexity. {N}otes of the {L}eningrad branch of
  the {S}teklov {M}athematical {I}nstitute, 1976.

\bibitem[GS11]{guruswami2011lasserre}
Venkatesan Guruswami and Ali~Kemal Sinop.
\newblock Lasserre hierarchy, higher eigenvalues, and approximation schemes for
  graph partitioning and quadratic integer programming with {PSD} objectives.
\newblock In {\em IEEE 52nd Annual Symposium on Foundations of Computer Science
  (FOCS)}, pages 482--491. IEEE, 2011.

\bibitem[GV18]{gillis2018complexity}
Nicolas Gillis and Stephen~A. Vavasis.
\newblock On the complexity of robust {PCA} and $\ell_1$-norm low-rank matrix
  approximation.
\newblock {\em Mathematics of Operations Research}, 43(4):1072--1084, 2018.

\bibitem[JLS{\etalchar{+}}21]{jiang2021single}
Yifei Jiang, Yi~Li, Yiming Sun, Jiaxin Wang, and David~P. Woodruff.
\newblock Single pass entrywise-transformed low rank approximation.
\newblock In {\em International Conference on Machine Learning}, pages
  4982--4991. PMLR, 2021.

\bibitem[Joh90]{johnson1990matrix}
Charles~R. Johnson.
\newblock Matrix completion problems: {A} survey.
\newblock In {\em Matrix theory and applications}, volume~40, pages 171--198,
  1990.

\bibitem[KMM23]{kim}
Eunjung Kim, {Arnaud de} Mesmay, and Tillmann Miltzow.
\newblock Representing matroids over the reals is {$\exists \mathbb R
  $}-complete.
\newblock {\em arXiv preprint arXiv:2301.03221}, 2023.

\bibitem[KMO10]{keshavan2010matrix}
Raghunandan~H. Keshavan, Andrea Montanari, and Sewoong Oh.
\newblock Matrix completion from a few entries.
\newblock {\em IEEE Transactions on Information Theory}, 56(6):2980--2998,
  2010.

\bibitem[KS09]{karpinski2009linear}
Marek Karpinski and Warren Schudy.
\newblock Linear time approximation schemes for the {G}ale-{B}erlekamp game and
  related minimization problems.
\newblock In {\em Proceedings of the 41st Annual ACM Symposium on Theory of
  Computing (STOC)}, pages 313--322, 2009.

\bibitem[KST54]{kst}
P.~K{\H{o}}v{\'a}ri, Vera~T. S{\'o}s, and Paul Tur{\'a}n.
\newblock On a problem of {Z}arankiewicz.
\newblock In {\em Colloquium Mathematicum}, volume~3, pages 50--57. Polska
  Akademia Nauk, 1954.

\bibitem[LW20]{li2020input}
Yi~Li and David~P. Woodruff.
\newblock Input-sparsity low rank approximation in {S}chatten norm.
\newblock In {\em International Conference on Machine Learning}, pages
  6001--6009. PMLR, 2020.

\bibitem[MM15]{manurangsi2015approximating}
Pasin Manurangsi and Dana Moshkovitz.
\newblock Approximating dense max 2-{CSP}s.
\newblock In {\em Approximation, Randomization, and Combinatorial Optimization.
  Algorithms and Techniques (APPROX/RANDOM)}. Schloss Dagstuhl-Leibniz-Zentrum
  fuer Informatik, 2015.

\bibitem[MMMN23]{voting}
Antoine M{\'e}ot, {Arnaud de} Mesmay, Moritz M{\"u}hlenthaler, and Alantha
  Newman.
\newblock Voting algorithms for unique games on complete graphs.
\newblock In {\em Symposium on Simplicity in Algorithms (SOSA)}, pages
  124--136. SIAM, 2023.

\bibitem[MR17]{manurangsi2017birthday}
Pasin Manurangsi and Prasad Raghavendra.
\newblock A birthday repetition theorem and complexity of approximating dense
  {CSP}s.
\newblock In {\em 44th International Colloquium on Automata, Languages, and
  Programming (ICALP)}. Schloss Dagstuhl-Leibniz-Zentrum fuer Informatik, 2017.

\bibitem[MS08]{mathieu2008yet}
Claire Mathieu and Warren Schudy.
\newblock Yet another algorithm for dense max cut: {G}o greedy.
\newblock In {\em Proceedings of the 19th Annual ACM-SIAM Symposium on Discrete
  Algorithms (SODA)}, pages 176--182. SIAM, 2008.

\bibitem[MW17]{musco2017sublinear}
Cameron Musco and David~P. Woodruff.
\newblock Sublinear time low-rank approximation of positive semidefinite
  matrices.
\newblock In {\em IEEE 58th Annual Symposium on Foundations of Computer Science
  (FOCS)}, pages 672--683. IEEE, 2017.

\bibitem[MW21]{mahankali2021optimal}
Arvind~V. Mahankali and David~P. Woodruff.
\newblock Optimal $\ell_1$ column subset selection and a fast {PTAS} for low
  rank approximation.
\newblock In {\em Proceedings of the 2021 ACM-SIAM Symposium on Discrete
  Algorithms (SODA)}, pages 560--578. SIAM, 2021.

\bibitem[Oxl06]{oxley}
James~G. Oxley.
\newblock {\em Matroid Theory}, volume~3.
\newblock Oxford University Press, USA, 2006.

\bibitem[Rec11]{recht2011simpler}
Benjamin Recht.
\newblock A simpler approach to matrix completion.
\newblock {\em Journal of Machine Learning Research}, 12(12), 2011.

\bibitem[RSW16]{razenshteyn2016weighted}
Ilya Razenshteyn, Zhao Song, and David~P. Woodruff.
\newblock Weighted low rank approximations with provable guarantees.
\newblock In {\em Proceedings of the 48th Annual ACM Symposium on Theory of
  Computing (STOC)}, pages 250--263, 2016.

\bibitem[SWZ17]{song2017low}
Zhao Song, David~P. Woodruff, and Peilin Zhong.
\newblock Low rank approximation with entrywise $\ell_1$-norm error.
\newblock In {\em Proceedings of the 49th Annual ACM Symposium on Theory of
  Computing (STOC)}, pages 688--701, 2017.

\bibitem[SWZ19]{song2019relative}
Zhao Song, David~P. Woodruff, and Peilin Zhong.
\newblock Relative error tensor low rank approximation.
\newblock In {\em Proceedings of the 30th Annual ACM-SIAM Symposium on Discrete
  Algorithms (SODA)}, pages 2772--2789. SIAM, 2019.

\bibitem[Val77]{valiant1977graph}
Leslie~G. Valiant.
\newblock Graph-theoretic arguments in low-level complexity.
\newblock In {\em Mathematical Foundations of Computer Science 1977:
  Proceedings, 6th Symposium, Tatransk{\'a} Lomnica September 5--9, 1977 6},
  pages 162--176. Springer, 1977.

\bibitem[Ver18]{vershynin2018high}
Roman Vershynin.
\newblock {\em High-dimensional probability: An introduction with applications
  in data science}, volume~47.
\newblock Cambridge University Press, 2018.

\bibitem[WY22]{woodruff2022improved}
David~P. Woodruff and Taisuke Yasuda.
\newblock Improved algorithms for low rank approximation from sparsity.
\newblock In {\em Proceedings of the 2022 Annual ACM-SIAM Symposium on Discrete
  Algorithms (SODA)}, pages 2358--2403. SIAM, 2022.

\bibitem[Yar14]{yaroslavtsev2014going}
Grigory Yaroslavtsev.
\newblock Going for speed: {S}ublinear algorithms for dense r-{CSP}s.
\newblock {\em arXiv preprint arXiv:1407.7887}, 2014.

\bibitem[YZ14]{yoshida2014approximation}
Yuichi Yoshida and Yuan Zhou.
\newblock Approximation schemes via {S}herali-{A}dams hierarchy for dense
  constraint satisfaction problems and assignment problems.
\newblock In {\em Proceedings of the 5th Conference on Innovations in
  Theoretical Computer Science}, pages 423--438, 2014.

\end{thebibliography}
